\documentclass[11pt]{article}
\usepackage[margin=1in]{geometry}
\usepackage{graphicx}
\usepackage{natbib}
\usepackage{footnote,enumerate,amsmath,amssymb,amsfonts,amsthm,subcaption,hyperref}
\usepackage[linesnumbered,ruled]{algorithm2e}

\newcommand{\M}{{\cal M}}
\newcommand{\PP}{{\cal P}}
\newtheorem{lemma}{Lemma}

\begin{document}
\title{Using Graph Partitioning for Scalable Distributed Quantum Molecular Dynamics}
\author{Hristo N.\ Djidjev\footnote{Los Alamos National Laboratory, Los Alamos, NM 87544, USA} \and Georg Hahn\footnote{Lancaster University, Bailrigg, Lancaster LA1 4YW, U.K.} \and Susan M.\ Mniszewski\footnotemark[1] \and Christian F.A.\ Negre\footnotemark[1] \and Anders M.N.\ Niklasson\footnotemark[1]}
\date{}
\maketitle

\begin{abstract}
The simulation of the physical movement of multi-body systems at an atomistic level, with forces calculated from a quantum mechanical description of the electrons, motivates a graph partitioning problem studied in this article. Several advanced algorithms relying on evaluations of matrix polynomials have been published in the literature for such simulations. We aim to use a special type of graph partitioning in order to efficiently parallelize these computations. For this, we create a graph representing the zero-nonzero structure of a thresholded density matrix, and partition that graph into several components. Each separate submatrix (corresponding to each subgraph) is then substituted into the matrix polynomial, and the result for the full matrix polynomial is reassembled at the end from the individual polynomials. This paper starts by introducing a rigorous definition as well as a mathematical justification of this partitioning problem. We assess the performance of several methods to compute graph partitions with respect to both the quality of the partitioning and their runtime.
\end{abstract}

\section{Introduction}
\label{section_intro}
The physical movements of multi-body systems on an atomistic level is at the core of molecular dynamics (MD) simulations. Those dynamics take place at the femtosecond ($10^{-15}$ second) time scale and they are incorporated in a larger simulation which typically is of the order of pico- to nanoseconds ($10^{-12}$ to $10^{-9}$ second). A simple way to conduct MD simulations is to derive all forces from the potential energy surface for all interacting particles, and to compute molecular trajectories for the multi-particle system by solving Newton's equations numerically. In quantum-based molecular dynamics (QMD) simulations, the electronic structure is based on an underlying quantum mechanical description, from which interatomic forces are calculated.

Several QMD methods are published in the literature for a variety of materials systems. So-called \textit{first principle} methods are capable of simulating a few hundred atoms over a picosecond range. Important approaches of this type include Hartree-Fock or density functional theory. Semiempirical methods such as self-consistent tight-binding techniques increase the applicability to systems of several thousand atoms. In contrast to regular first principle methods, approximate methods are often two to three orders of magnitude faster while still capturing the quantum mechanical behavior (for example, charge transfer, bond formation, excitations, and quantum size effects) of the system.

One of the most efficient and widely used methods is density functional based self-consistent tight-binding theory \citep{Elstner1998, MFinnis98, TFrauenheim00}. In this approach, the main computational effort stems from the diagonalization of a matrix, the so-called \textit{Hamiltonian matrix}, which encodes the electronic energy of the system. The Hamiltonian matrix is needed in order to construct the \textit{density matrix} describing the electronic structure of the system. Before evaluating the forces at each time step of a QMD simulation, the self-consistent construction of the density marix is carried out. Computing the density matrix requires matrix diagonalization with a computational cost of $O(N^3)$, where $N$ is the dimension of the Hamiltonian matrix. This makes diagonalization only viable for small system sizes. For this reason, the last two decades have seen the development of a number of linear runtime (i.e., $O(N)$) algorithms.

One such linear runtime approach is based on a recursive polynomial expansion of the density matrix \citep{Niklasson2002}. A linear scaling with the system size for non-metallic systems is achieved by the sparse-matrix second-order spectral projection (SM-SP2) algorithm. On dense or sparse matrices, SM-SP2 competes with or outperforms regular diagonalization schemes with respect to both speed and accuracy \citep{Mniszewski2015}. SM-SP2 uses the expression
\begin{equation}
  D=\lim\limits_{i \rightarrow \infty } f_i[f_{i-1}[\dots f_0[X_0]\dots]]
  \label{eq:expansion}
\end{equation}
to compute the density matrix $D$ from the Hamiltonian $H$, where $f_i(X_i)$ is a quadratic function (either $X_i^2$ or $2X_i-X_i^2$, depending on $Tr(X_i)$ or $Tr(X_{i+1})$), and the initial matrix $X_0$ is a linearly modified version of $H$. Usually, $20$-$30$ iterations suffice to obtain a close approximation of $D$. Additionally, thresholding is applied to further reduce the computational complexity, where small nonzero elements of the matrix (typically between $10^{-5}$ to $10^{-7}$) are set to zero. Thresholding the matrix elements of the Hamiltonian is essential in order to arrive at a fast way to calculate the density matrix: without thresholding, the resulting matrix would quickly become dense due to fill-in.

The cost of computing a matrix polynomial $P$ (which is mainly due to the squaring of a matrix) dominates the computational cost of the SM-SP2 algorithm. Since we are interested in performing a large number of time steps (of the order of $10^4$-$10^6$) in a typical QMD simulation, we need to parallelize the evaluation of the matrix polynomials in order to keep the wall-clock time low. However, the significant communication overhead for every iteration causes SM-SP2 to not parallelize well. Linear scaling complexity has been achieved with thresholded blocked sparse matrix algebra \citep{Bock,Hutter,Mniszewski2015,VandeVondele}. In our paper, we present an alternative formulation that reduces communication overhead via graph partitioning and enables scalable parallelism. Our basic approach was introduced in \cite{Niklasson2016}, but with the main focus being on the physics aspects.

This paper addresses the aforementioned parallel version of SP2 together with an inbuilt partitioning scheme applied to the graph representation of the density matrix, denoted as G-SP2 in the remainder of the article. In particular, the computational aspects of evaluating the matrix polynomial in G-SP2 are investigated. We represent the Hamiltonian (or density) matrix as a graph where atomic orbitals are given as vertices and non-zero interactions become edges, and then partition that graph into blocks (parts) with the aim to minimize a suitable cost function. The submatrices of the Hamiltonian or density matrix correspond to the divisions of the molecule, and are derived from the graph partition. We rigorously prove that applying the matrix polynomial to the entire Hamiltonian is equivalent to applying it to the partitioned Hamiltonian matrix and re-assembling the partial solutions, thus justifying our approach to parallelize the computational workload via graph partitioning. Although demonstrated numerically \citep{Niklasson2016}, no such proof exists in the literature to the best of our knowledge.

Our algorithm is related to the general framework of \cite{PinarHendrickson2001} (specifically the case of \textit{overlapped subdomains}), who present a general methodology to partition workloads among several processors. Their work is similar in that standard graph partitionings are considered, which are then locally modified to further optimize them (we attempt the same with a simulated annealing approach in Section~\ref{section_SA}). However, there are crucial differences to the work presented in this article: Since their framework is very general, it does not consider our specific objective function presented later in Section~\ref{section_theory}, their general algorithm works with a distance cutoff for modifying partitions which we do not recommend, and most importantly, we prove that for our problem, the partitioned matrix polynomial indeed allows lossless parallelization.

We show that partitions with overlapping blocks (or \textit{halos}) need to be computed in order to obtain accurate results with our graph approach. Naturally, the computational overhead increases with the overlap. The aim of this work is to minimize the computational cost of the matrix polynomial evaluation through appropriate partitioning schemes. Those schemes will directly minimize the cost of the corresponding polynomial evaluation as opposed to traditional edge cut minimization. In our article, we will experimentally study several algorithms for the aforementioned graph partitioning problem, which we formally introduce first. In \cite{Looz2016}, the authors consider a similar set-up in which the graph encoding the molecular structure is partitioned (using various techniques) for faster heuristic computation; however, their article crucially differs from ours in that there are no generally applicable results pertaining to lossless partitioning of the matrix polynomial using overlapping blocks.

Our approach allows us to avoid communication between processors after each iteration of \eqref{eq:expansion} until the entire polynomial is evaluated. To this end, each processor independently evaluates its assigned polynomial after partitioning and distributing the inital matrix, and the final output is assembled from the computed submatrices. In this article, different algorithmic approaches for computing graph partitions are assessed with respect to the aforementioned objective function and their computational effort. Importantly, we analyze the tradeoff between the additional computational costs for computing graph partitions before running the SP2 algorithm in parallel, and carrying out regular molecular dynamics. We also investigate the tuning of the number of blocks as a function of the graph size in order to minimize computational effort.

The structure of this paper is as follows. Section~\ref{section_theory} introduces the mathematical foundations for partitioning the evaluation of matrix polynomials, states our algorithm including a proof of correctness, and defines the graph partitioning problem we consider. Algorithms for constructing such partitions and their implementations are discussed in Section~\ref{section_algorithms}. Experimental results for several physical test systems are given in Section~\ref{section_results}. Section~\ref{section_discussion} discusses our results. Proofs for Section~\ref{section_theory} can be found in Appendix~\ref{appendix_proofs}, and Appendix~\ref{app:experiments} presents further experimental results.

A preliminary version of this article has been published as a conference paper in the \textit{SIAM Workshop on Combinatorial Scientific Computing (CSC16)}, see \cite{Djidjev2016}. This article is an extension of the work of \cite{Djidjev2016}, and includes proofs of all theoretical results of Section~\ref{section_theory} in Appendix~\ref{appendix_proofs}, pseudo-code of our simulated annealing approach in Section~\ref{section_SA}, a visualization of the relationship between the graph structure of a molecule and its partitioned graph representation in Section~\ref{subsection_assessment_all}, and more detailed performance data used for all experiments in Appendix~\ref{app:experiments}.

\section{Evaluating Matrix Polynomials on Partitions}
\label{section_theory}
We define a thresholded matrix polynomial, justify its parallelized evaluation, present an algorithm to evaluate a matrix polynomial in a parallelized fashion, and conclude by defining the cost function for an implied graph partitioning problem.

We encode the zero-nonzero structure of a symmetric matrix $X=\{x_{ij}\}$ as a graph $G(X)$, called the \textit{sparsity graph} of $X$. $G(X)$ contains a vertex for each row (or column) in $X$, and $G(X)$ contains an edge between vertices $i$ and $j$ if and only if $x_{ij} \neq 0$. We now generalize the matrix polynomial defined in \eqref{eq:expansion} for any symmetric $n \times n$ matrix $A$. Denote the superposition of operators of the type
\begin{equation}
  P=P_1\mathsmaller{\circ} T_1\mathsmaller{\circ}\dots\mathsmaller{\circ} P_s\mathsmaller{\circ} T_s
  \label{eq:poly2}
\end{equation}
as a \textit{thresholded matrix polynomial} of degree $2^s$, where $T_i$ is a thresholding operation and $P_i$ is a polynomial of degree $2$. For any graph $I$, we formally define $T_i$ as the graph operator (with associated edge set $E(T_i)$) such that $T_i(I)$ is a graph with a vertex set $V(I)$ and an edge set $E(I) \setminus E(T_i)$.

The application of a superpositioned operator $P$ of the type \eqref{eq:poly2} to a matrix $A$ of appropriate dimension is denoted as $P(A)$. In \eqref{eq:poly2}, $P$ is composed of polynomials $P_i$ and thresholding operations $T_i$. For our SM-SP2 application, the Hamiltonian is $A$ and the density matrix is $P(A)$.

For any matrix $A$, we define all matrices $B$ which have the same zero-nonzero structure as $A$ (that is, $G(A)=G(B)$) to be in the structure class $\M(A)$. We observe that the non-zero structure of $P(A)$ and $P(B)$ can be different.

Let $G=G(A)$ for a matrix $A$ and let $P$ be a thresholded matrix polynomial. We denote by $\PP(G)$ the minimal graph with the same vertices as $G$ such that if $P(B)|_{vw}\neq 0$ for any matrix $B \in \M(A)$ and any $v,w$, then there is an edge $(v,w)\in E(\PP(G))$. We interpret $\PP(G)$ as the worst-case zero-nonzero structure of $P(A)$ that excludes cancellations resulting from the addition of opposite-sign numbers, thus resulting in coincidental zeros. All diagonal elements of $A$ are assumed to be non-zero, and $E(T_i)$ is assumed to not contain a loop edge.

Assume we are given a set $\Pi=\{\Pi_1,\dots,\Pi_q\}$, where each $\Pi_i = U_i \cup W_i$, called a \textit{block} of $\Pi$, is a union of a \textit{core} vertex set $U_i$ and a \textit{halo} vertex set $W_i$. Given the two following conditions hold true, we call $\Pi$ a \textit{CH-partition} (or core-halo partition):
\begin{enumerate}
  \item $\bigcup_iU_i=V(G),\;U_i\cap U_j=\emptyset$ for all $i\neq j$;
  \item neighbors of vertices in $U_i$ that are themselves not in $U_i$ are contained in $W_i$.
\end{enumerate}

Let $H_{U_i}$ be the subgraph of $H=\PP(G)$ induced by all neighbors of $U_i$ in $H$. We combine all rows and columns of $A$ that correspond to vertices of $V(H_{U_i})$ in a submatrix $A_{U_i}$ of $A$. The following lemma shows that $P(A)$ can be computed on submatrices of the Hamiltonian.

\begin{lemma}\label{lem:matrixValuesMulti2}
  For any $v\in U_i$ and any neighbor $w$ of $v$ in $\PP(G)$, the element of $P(A)$ corresponding to the edge $(v,w)$ of $\PP(G)$ is equal to the element of $P(A_{U_i})$ corresponding to the edge $(v,w)$ of $H_i$.
\end{lemma}

The proof is given in Appendix~\ref{appendix_proofs}.
Lemma~\ref{lem:matrixValuesMulti2} justifies the parallelized evaluation of a matrix polynomial.

Let us apply the aforementioned results to a Hamiltonian matrix $A$ in a QMD simulation. (See the example in Figure~\ref{fig:example}). In this case, we can assume that the sparsity structure of the density matrix $D$ from the previous QMD simulation step is being passed on to $P(A)$. We can thus approximate $H=\PP(G)$ with $G(D)$ since the graph $H$ is unknown until $P(A)$ is computed. This is a reasonable approximation since the connectivity will only change marginally from step to step in the QMD simulation. Since by construction, the graph $H$ is totally contained in the graph of its density matrix (which is the object we want to compute and thus unknown), it is sensible to approximate its connectivity from the density matrix of the previous time step of the QMD simulation. The halos can also be taken from $H$ in practice.

\begin{figure}
	\begin{subfigure}{6.5cm}\centering
	$~~~A= \left(\begin{matrix}-1.2&1.89&0&0&0\\1.89&0.92&0.08&0&0\\0&0.08&0.85&0.11&0\\0&0&0.11&0.78&1.21\\0&0&0&1.21&-1.31\end{matrix}\right)$
	\caption{}
  \end{subfigure}
	\begin{subfigure}{8cm}\centering
		$P(A)=\left(\begin{matrix}5.01&-0.53&0.15&0.&0.\\-0.53&4.43&0.14&\textbf{0.009}&0.\\0.15&0.14&0.74&0.18&0.13\\0.&\textbf{0.009}&0.18&2.08&-0.64\\0.&0.&0.13&-0.64&3.18\end{matrix}\right)$
	\caption{}	\end{subfigure}
	\begin{subfigure}[b]{2.8cm}\centering
	\includegraphics[width=0.95\linewidth]{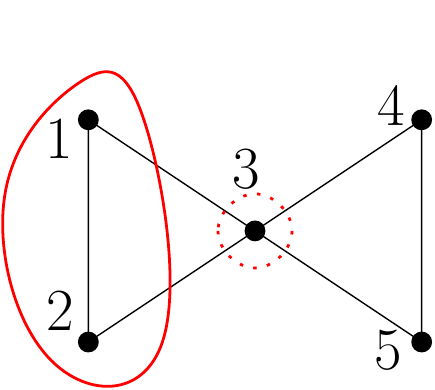}
	\caption{}  \end{subfigure}\hfill
  \begin{subfigure}[b]{5cm}\centering
	$A_U= \left(\begin{matrix}-1.2&1.89&0\\1.89&0.92&0.08\\0&0.08&0.85\end{matrix}\right)$
	\caption{}  \end{subfigure}\hfill
  \begin{subfigure}[b]{5.5cm}\vfill\centering
	$P(A_U)=\left(\begin{matrix}5.01&-0.53&0.15\\-0.53&4.43&0.14\\0.15&0.14&0.73\end{matrix}\right) $
	\caption{}\end{subfigure}
	\caption{Example of using CH-partitioning for distributed evaluation of matrix polynomial $P(A)=A^2$, with a thresholding operator that removes all edges/matrix entries with values below $0.01$. (a) Input symmetric matrix $A$. (b) $P(A)=A^2$. (c) After the thresholding operator removes elements $(2,4)$ and $(4,2)$ (in bold) from $P(A)$,  the sparsity graph $\PP(G)$ of the resulting matrix is shown. We are considering a block of a CH-partitioning with core set   $U=\{1,2\}$ (encircled with a solid line) and halo the set $H=\{3\}$ (encircled with a dashed line). (d) Submatrix $A_U$ consisting of the first three (by $U\cup H=\{1,2,3\}$) rows and columns of $A$. (e) Matrix $P(A_U)=A_U^2$. According to Lemma~\ref{lem:matrixValuesMulti2}, the elements of the first two (by $U=\{1,2\}$) columns/rows of $P(A_U)$ are equal to the corresponding nonzero elements of the first two columns/rows of thresholded $P(A)$.}
	\label{fig:example}
\end{figure}

We propose the following algorithm for computing $P(A)$ from $H=G(D)$:
\begin{enumerate}
  \item Divide $V(G)$ into $q$ disjoint sets $\{U_1,\dots,U_q\}$ and define a CH-partition $\Pi=\{\Pi_1,\dots,\Pi_q\}$, where $\Pi_i$ has core $U_i$ and halo $N(U_i,H)\setminus U_i$;
  \item Construct submatrices $A_{U_i}$ for all $i=1,\dots,q$;
  \item Compute $P(A_{U_i})$ for all $i$ independently using dense matrix algebra;
  \item Given a vertex $i$, let $k$ be the index such that the set $U_k$ contains $i$. Let $j$ be the row in $A_{U_k}$ that corresponds to the $i$-th row in $A$. Then, define $P(A)$ as a matrix whose $i$-th row equals the $j$-th row of $P(A_{U_k})$.
\end{enumerate}
This algorithm computes $P(A)$ as demonstrated in Lemma~\ref{lem:matrixValuesMulti2}.

The computational bottleneck of the algorithm for $P(A)$ is caused by the dense matrix-matrix multiplication required to compute $P(A_{U_i})$ for all $i$ in step~(3).

To be precise, according to \eqref{eq:poly2}, computing $P(A_{U_i})$ takes $s(c_i+h_i)^3$ operations, where $c_i$ and $h_i$ are the size of the core and the halo of $\Pi_i$ and $s$ is the number of superpositioned operators (see \eqref{eq:poly2}). In this calculation, the computational effort for thresholding some matrix elements is excluded, since this effort is quadratic in the worst case and linear in $c_i+h_i$ in average cases. We observe that a CH-partitioning which minimizes the effort to compute $P(A)$ also minimizes $\sum_{i=1}^q (c_i+h_i)^3$ due to the fact that $s$ is independent of $\Pi$.

This observation motivates our \textit{CH-partitioning problem}, defined as follows: For an undirected graph $G$ and an integer $q \geq 2$, split $G$ into $q$ blocks $\Pi_1,\dots,\Pi_q$ such that
\begin{equation}
  \label{eq:sumOfCubes}
  \sum_{i=1}^q (c_i+h_i)^3
\end{equation}
is minimized, where $\Pi_i$ has a core $U_i$ of size $c_i$ and a halo $N(U_i,G) \setminus U_i$ of size $h_i$.

As an example, the optimal CH-partitioning for a star graph of $n$ vertices has a single non-empty block containing all vertices: this block is composed of the central vertex, whose halo contains all other vertices. In contrast, a standard (edge cut) partitioning will have $n$ parts, thus demonstrating that a CH-partitioning minimizing \eqref{eq:sumOfCubes} can be quite different from a standard balanced partitioning.

\section{Algorithms for Graph Partitioning Considered in our Study}
\label{section_algorithms}
In this section we investigate the ability of existing graph partitioning packages, as well as our own heuristic algorithm, for computing CH-partitions that minimize the objective function \eqref{eq:sumOfCubes}. Those algorithms are: \textit{METIS} (version 5.1.0) and \textit{hMETIS} (version 1.5) due to their widespread use, and \textit{KaHIP} (version 0.7) based on its convincing performance at the 10th DIMACS Implementation Challenge~\cite{DBLP:conf/dimacs/2012}.

\subsection{Edge Cut Graph Partitioning}
We observe that we obtain $|V(G)| + \sum_i h_i$ when leaving out the cubes in the objective function \eqref{eq:sumOfCubes}, thus making it necessary to minimize the sum of the halo nodes over all blocks.

Regular graph partitions and CH-partitions are related. Suppose we are given a regular partition $P$. For any part in $P$, we can define a core corresponding to that part, and a halo consisting of all adjacent vertices of the core vertices (excluding the core vertices themselves). We define the CH-partition $\Pi$ to consist of precisely those blocks for any element of $P$. It must then be true that either $v$ or $w$ is a halo vertex for any cut edge $(v,w)$ of $P$. Conversely, there exists a core vertex $w$ such that $(v,w)$ is a cut edge for any halo vertex $v$ belonging to some part in $\Pi$.

This shows that the cut edges of $P$ and the set of halo nodes in $\Pi$ are related but not equal. We observe that another measure, the \textit{total communication volume}, exactly corresponds to the sum of halo nodes. Certain tools like \textit{METIS} allow us to optimize with respect to the total communication volume.
By ignoring the cubes in \eqref{eq:sumOfCubes}, we aim to study how well CH-partitions can be produced by regular graph partitioning tools.
Additionally, we improve the solutions obtained by standard graph partitioning tools with our own heuristic in Section~\ref{section_SA}.
The three following algorithms will be used:

\subsubsection{METIS}
\label{subsection_metis}
\textit{METIS} \citep{KarypisKumar1999} uses a three-phase multilevel approach to perform graph partitioning:
\begin{enumerate}
  \item Starting from the original graph $G=G_n$ (where $n=|V|$), \textit{METIS} generates a sequence of graphs $G_n, G_{n-1}, \ldots, G_{n'}$ for some $n'<n$ to coarsen the input graph. The coarsening ends with a suitably small graph $G_{n'}$ (typically having $n'<100$ vertices).
  \item An algorithm of choice is used to partition $G_{n'}$.
  \item Using the sequence $G_{n'},\dots,G_n$, the partitions are expanded back from $G_{n'}$ to the full graph $G_n$.
\end{enumerate}
Additionally, \textit{METIS} employs a refinement algorithm such as the one of Fiduccia-Mattheyses \citep{FiducciaMattheyses1982} to improve the partitioning after each projection. This is necessary since the finer the partition, the more degrees of freedom it has during the uncoarsening phase.
Several tuning parameters can be set in \textit{METIS}, including the size of $G_n$, the coarsening algorithm, and the algorithm used for partitioning $G_n$.

\subsubsection{KaHIP}
Several multilevel graph partitioning algorithms are combined in \textit{KaHIP}~\cite{sandersschulz2013}.
\textit{KaHIP} works similarly to \textit{METIS}. A given input graph is first contracted, partitions are computed on the coarsest contraction level, and the partitions obtained in this way are projected back to coarser levels, where refinement algorithms are used to enhance the solutions. \textit{KaHIP} offers max-flow/min-cut \citep{SandersSchulz2011,FordFulkerson1956} and Fiduccia-Mattheyses refinement \citep{FiducciaMattheyses1982} for local improvement of the solution. Additionally, F-cycles \citep{SandersSchulz2011} can be employed as a global search technique.

\subsubsection{Hypergraph partitioning}
\label{subsection_hmetis}
A hypergraph formulation is an alternative approach to the classical interpretation of the density matrix as an adjacency matrix (where each nonzero interaction in the density matrix is an undirected and unweighted edge).

The set of all neighbors of a vertex form a single hyperedge in the hypergraph formulation. The main advantage of using hyperedges consists in the fact that minimizing the edge-cut with respect to hyperedges results in either all or zero vertices being included in a part: if a hyperedge is not a cut-edge, then this automatically implies that all the neighbors of that core vertex (i.e., its halo) are included in the same part.

We use \textit{hMETIS} of \cite{KarypisKumar2000} in order to compute hypergraph partitionings. \textit{hMETIS} is the hypergraph analog of \textit{METIS}.

\subsection{Refinement with Simulated Annealing}
\label{section_SA}
The objective function \eqref{eq:sumOfCubes} we aim to minimize differs from the size of the edge or hyperedge cut minimized by standard graph and hypergraph partitioning algorithms. With the help of the algorithm derived in this section we explicitly minimize \eqref{eq:sumOfCubes}.

A standard tool in optimization is the probabilistic algorithm of \citep{Kirkpatrick1983}, called simulated annealing (\textit{SA}). SA iteratively proposes random modifications to an existing solution in order to improve it, and thus optimizes without gradients. If a proposed modification (or move) does not immediately lower the objective function, it might still be accepted with a certain \textit{acceptance probability}. The acceptance probability is proportional to the magnitude of the (unfavorable) increase in the objective function, and antiproportional to the runtime. The latter makes it more likely for SA to accept unfavorable moves in the exploration phase at the start of each run, and it is implemented using a strictly decreasing \textit{temperature} function. Modifications to the existing solution which further minimize the objective function are always accepted.

A fixed number of iterations, a vanishingly small temperature, or the lack of further improvements in the solution over a certain number of iterations can be employed as stopping criteria for SA.

We test the following proposal functions that return modifications to subgraphs induced by an existing block of the partition, where $P$ and $P'$ are simply sets of nodes:
\begin{enumerate}
  \item Select a random block $P'$, select one of its halo nodes $v$ at random and move $v$ into block $P$.\label{sa_item_halo}
  \item Select a random block $P'$, select one of its nodes $v$ at random and move $v$ into $P$.\label{sa_item_random}
  \item Select the block $P'$ with most halo nodes and (a) move a random node $v$ into $P$, (b) make a random halo node of $P$ a core node, or (c) move any node of $P$ to another block.\label{sa_item_largest}
  \item Like (\ref{sa_item_largest}.) using the block $P'$ with the largest sum of core and halo nodes.
\end{enumerate}
Many more sensible proposal functions could be devised. However, in our experiments we observed that the above proposals result in a similar behavior of SA, with the best tradeoff between speed and performance being achieved by scheme (\ref{sa_item_halo}.).

\begin{algorithm}[t]
  \SetKwRepeat{Repeat}{Repeat}{x}
  \DontPrintSemicolon
  \caption{Simulated Annealing\label{algorithm_sa}}
  \KwIn{Graph G, number of iterations $N$, initial partitioning $\Pi$}
  \KwOut{Updated partitioning $\Pi$}
  Select a temperature function $t(i)=1/i$ ;\\
  \For{$i=1$ to $N$}{
    Select random block $\pi$ in $\Pi$ and a random core-halo edge $(v,w)$ in $\pi$;\\
    Make $w$ a core vertex of $\pi$ and update the halo of $\pi$;\\
    Compute the values $S$ and $S'$ of \eqref{eq:sumOfCubes} for $\Pi$ and $\Pi'$, respectively, and set $\Delta=S'-S$;\\
    Compute $p = \exp \left(-\Delta/t(i) \right)$;\\
    Set $\Pi=\Pi'$ with probability $\min(1,p)$;\\
  }
  Output partitioning $\Pi$;\\
\end{algorithm}

Algorithm~\ref{algorithm_sa} states the SA implementation we use in our experiments. Any regular (edge cut) partitioning (e.g., obtained with \textit{METIS}), and even a random partitioning, serves as input $\Pi$ (the set of blocks) to Algorithm~\ref{algorithm_sa}. The SA algorithm runs over a fixed number of $N$ iterations. In each iteration, we randomly select a block $\pi \in \Pi$ as well as a random edge joining a core vertex $v$ with a halo vertex $w$ in $\pi$. Afterwards, $\pi$ is updated with $w$ being a core vertex. After storing the new partitioning in a set $\Pi'$, both $\Pi$ and $\Pi'$ are assessed. For this we compute the change $\Delta$ in the objective function \eqref{eq:sumOfCubes} between $\Pi'$ and $\Pi$, which is used to update the acceptance probability $p$. We accept $\Pi'$ (thus overwriting $\Pi:=\Pi'$) with probability $p$. Note that $p>1$ if $\Delta<0$, thus leading to a guaranteed acceptance of a proposal that directly improves \eqref{eq:sumOfCubes}. Afterwards, the iteration is repeated.

SA is run with a maximal number of iterations $N=100$ as stopping criterion, and the temperature function is chosen as $t(i)=1/i$. In practice, we first threshold a density matrix (see Section~\ref{section_theory}), convert it to a graph, compute regular edge cut partitionings for the converted graph with a standard software, and post-process the partition with SA.

\section{Experiments}
\label{section_results}
Using three measures the quality of the CH-partitions returned by the algorithms of Section~\ref{section_algorithms} is evaluated. Those measures are the objective function (\ref{eq:sumOfCubes}), the algorithmic runtime, and the number of MPI ranks and threads in an assessment of the scaling behavior of the G-SP2 algorithm (for one fixed system). We employ graphs derived from representations of actual molecules as opposed to simulated random graphs.

Unless noted otherwise, all experiments were performed on the Darwin cluster of Los Alamos National Laboratory, which hosts 20-core 2.60 GHz Intel Xeon E5-2660 processors with v3 architecture.

\subsection{Parameter Choices for METIS and hMETIS}
\label{subsection_parameters}
Using a grid search over sensible values, we tune the parameters of \textit{METIS} and \textit{hMETIS}, and will keep the set of parameters yielding the best average performance for all systems considered in this section fixed throughout the remainder of the simulations.

The default multilevel $k$-way partitioning as well as the default sorted heavy-edge matching for coarsening the graph were employed to run \textit{METIS}. Importantly, the user can choose to minimize either the \textit{edge cut} or the \textit{total communication volume} of the partitioning within the $k$-way partitioning routine of \textit{METIS}. As our definition of the \textit{sum of halo nodes} in Section~\ref{section_theory} is equivalent to the definition of the total communication volume of \textit{METIS}, we choose this option.

We employ \textit{hMETIS} with the followed parameters: we use recursive bisectioning (instead of $k$-way partitioning) with vertex grouping scheme \textit{Ctype}$=1$ (meaning the hybrid first-choice scheme HFC), Fiduccia-Mattheyses refinement (refinement heuristic parameter set to $1$), and $V$-cycle refinement on each intermediate solution ($V$-cycle refinement parameter set to $3$). These parameters are explained in the \textit{hMetis} manual~\cite{hmetis_manual}.

\subsection{A Collection of Test Graphs Derived from Molecular Systems}
\label{subsection_systems}
This section uses a selection of physical test systems to evaluate all algorithms of Section~\ref{section_algorithms}, which were chosen to represent a variety of realistic scenarios where graph partitioning can be applied to MD simulations. Those test systems are available from the authors upon request. We demonstrate how the graph structure influences the results (Section~\ref{subsection_assessment_all}) and additionally provide insights into the physics of each test system.

\begin{figure}
  \begin{center}
    \includegraphics[width=0.25\textwidth]{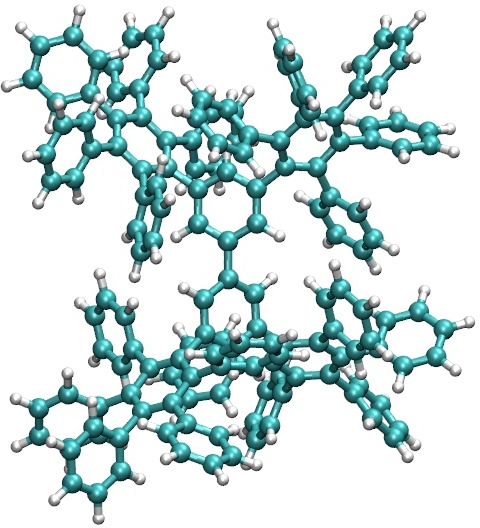}~
    \includegraphics[width=0.30\textwidth]{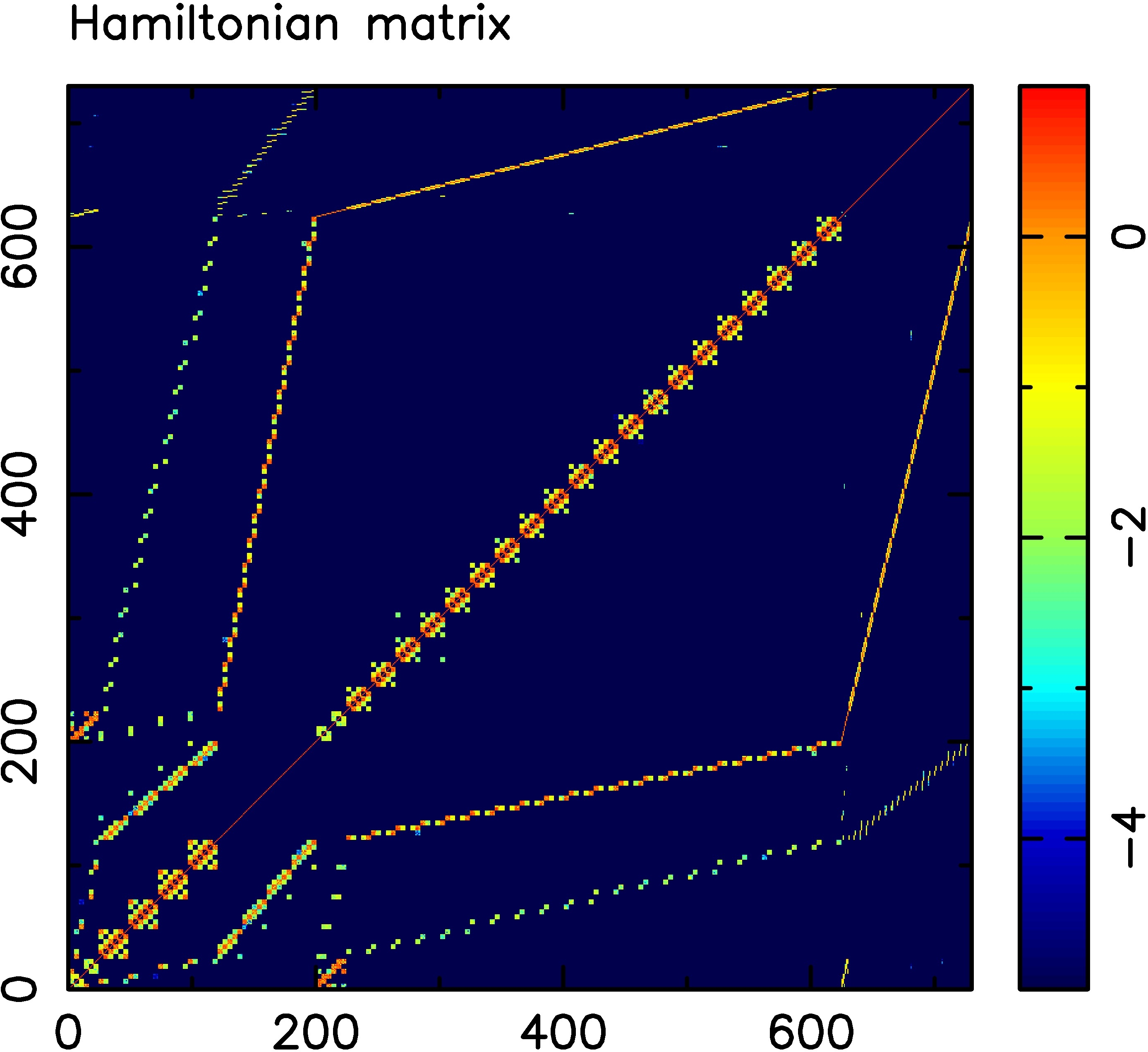}~
    \includegraphics[width=0.30\textwidth]{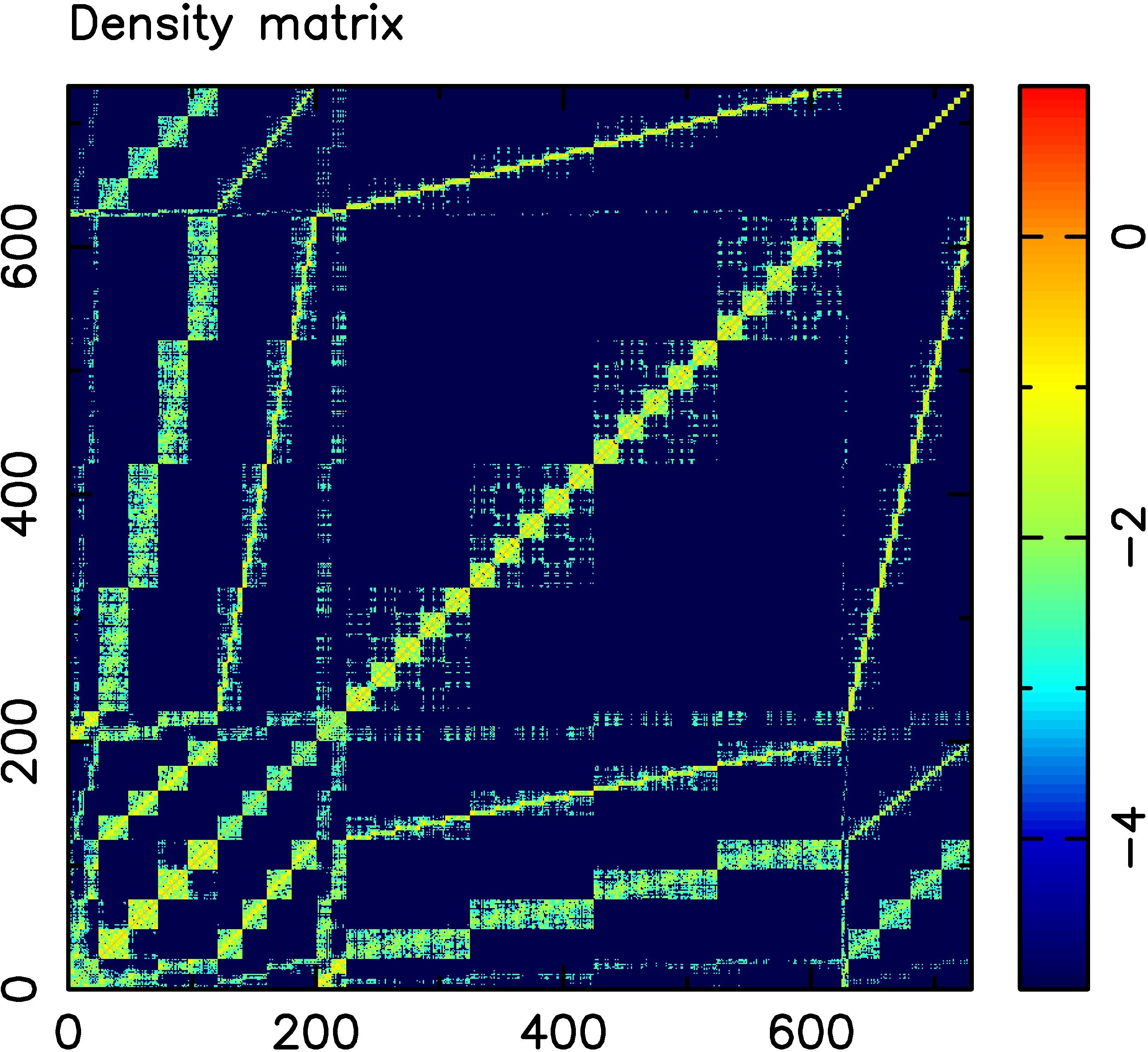}
    \caption{Molecular representation of phenyl dendrimer (left) using cyan and white spheres for carbon and hydrogen atoms, respectively. Hamiltonian matrix as 2D representation (middle) and thresholded density matrix (right). All plots show log$_{10}$ of the absolute values of all matrix elements. The SM-SP2 algorithm was used to compute the density matrix. Figure taken from \cite{Djidjev2016}. Copyright \textcopyright 2016 Society for Industrial and Applied Mathematics. Reprinted with permission. All rights reserved.\label{fig:dendrimer_molecule}}
  \end{center}
\end{figure}

A dendrimer molecule with 22 covalently bonded phenyl groups of solely C and H atoms is schematically shown in Figure~\ref{fig:dendrimer_molecule} (left). The graph of the dendrimer molecule has 730 vertices, composed of 262 atoms and 730 orbitals.

Figure~\ref{fig:dendrimer_molecule} (middle) displays the absolute values of the Hamiltonian matrix for the dendrimer system. Figure~\ref{fig:dendrimer_molecule} (right) displays the density matrix encoding the physical properties of the system, which is obtained by applying the SM-SP2 algorithm to the Hamiltonian.

We threshold the density matrix at $10^{-5}$ to convert it into a graph needed to find meaningful physical components via graph partitioning.
This is done with all systems of Table~\ref{tab:systems} to arrive at their adjacency matrices.
\begin{table}[t]
  \caption{Physical systems of our study: number of vertices $n$ in the graph and number of edges $m$. Table taken from \cite{Djidjev2016}. Copyright \textcopyright 2016 Society for Industrial and Applied Mathematics. Reprinted with permission. All rights reserved.\label{tab:systems}}
  \begin{center}
	\scriptsize
    \begin{tabular}{|l|r|r|r|l|}
      \hline
      Name & $n$ & $m$ & $m/n$ & Description\\
      \hline
      polyethylene dense crystal	& 18432	& 4112189	& 223.1	& crystal molecule in water solvent (low threshold)\\
      polyethylene sparse crystal	& 18432	& 812343	& 44.1	& crystal molecule in water solvent (high threshold)\\
      phenyl dendrimer 			& 730 	& 31147 	& 42.7	& polyphenylene branched molecule\\
      polyalanine 189 			& 31941	& 1879751	& 58.9	& poly-alanine protein solvated in water\\
      peptide 1aft 			& 385 	& 1833 		& 4.76	& ribonucleoside-diphosphate reductase protein\\
      polyethylene chain 1024 		& 12288	& 290816 	& 23.7	& chain of polymer molecule, almost 1-dimensional\\
      polyalanine 289 			& 41185	& 1827256	& 44.4	& large protein in water solvent\\
      peptide trp cage 			& 16863	& 176300	& 10.5	& smallest protein with ability to fold (in water)\\
      urea crystal			& 3584	& 109067	& 30.4	& organic compound in living organisms\\
      \hline
    \end{tabular}
  \end{center}
\end{table}
The first column of Table~\ref{tab:systems} displays the molecule name, the number of vertices $n$ and edges $m$ (second and third column) of its graph representation, and its average vertex degree $m/n$ (fourth column). A short description of the molecule can be found in the last column.

\begin{figure}
  \begin{center}
    \includegraphics[width=0.22\textwidth]{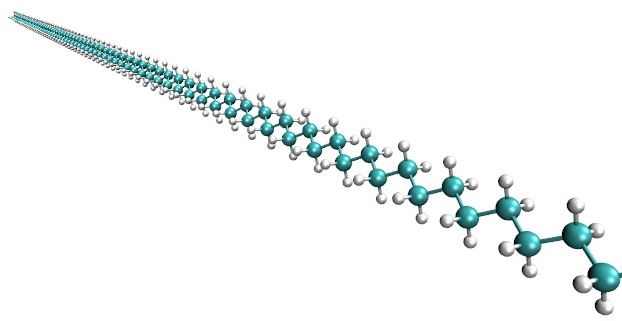}
    \includegraphics[width=0.22\textwidth]{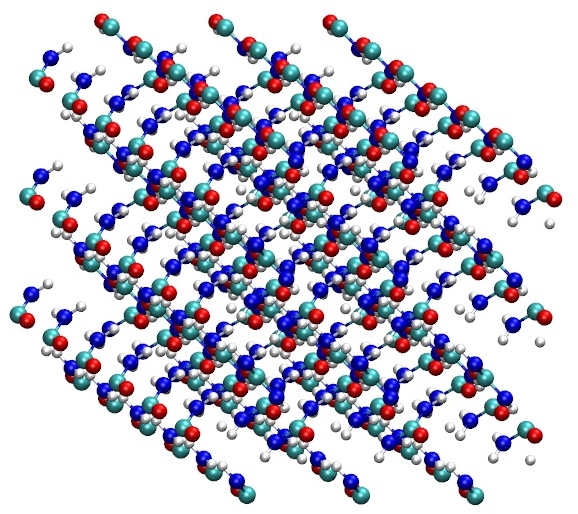}
    \includegraphics[width=0.22\textwidth]{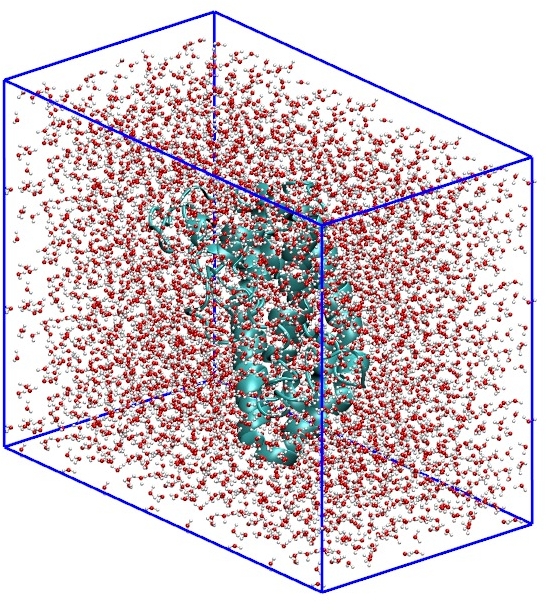}
    \includegraphics[width=0.22\textwidth]{dendrimer1.jpg}
    \caption{Molecular systems of this study: polyethyene linear chain (first plot), urea crystal (second plot), 189 residue polyalanine solvated in a water box (third plot), and phenyl dendrimer molecule (fourth plot). Cyan, blue, red and white spheres represent carbon, nitrogen, oxygen and hydrogen atoms, respectively. Figure taken from \cite{Djidjev2016}. Copyright \textcopyright 2016 Society for Industrial and Applied Mathematics. Reprinted with permission. All rights reserved.\label{fig:systems}}
  \end{center} 
\end{figure}

We consider four topologically different types of molecular systems (see Figure~\ref{fig:systems}) as test beds for the partitioning algorithms. This is to provide a wider range of test systems for our algorithms. Figure~\ref{fig:systems} displays a one dimensional system in the first panel (polyethyene linear chain with repeated CH$_{2}$ units), an anisotropic pristine 3D urea crystal (second panel), a polyalanine molecule solvated in water (third panel) with a typical $\alpha$-helix secondary structure, and a dendrimeric system with a fractal arrangement of phenyl rings (fourth panel) to challenge our partitioning algorithms.

\subsection{Comparison of the Partitioning Algorithms}
\label{subsection_assessment_all}
\begin{figure}[t]
  \centering
  \includegraphics[width=0.9\textwidth]{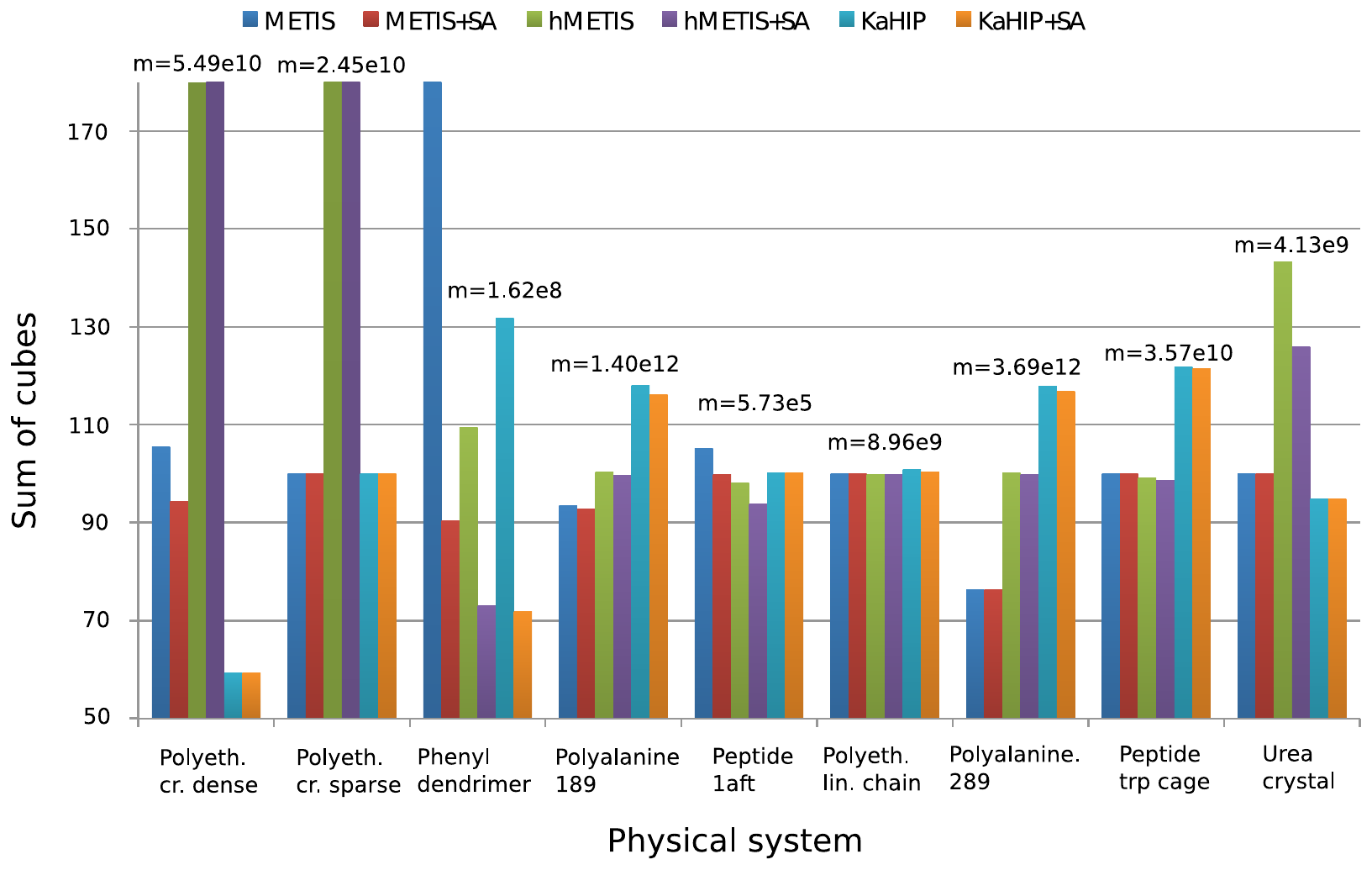}
  \caption{Sum of cubes performance measure to evaluate partitions. Values are normalized to have a median of $100$. To make the chart more informative, very large values are truncated, though all exact values can be found in \cite[Table~$2$]{Djidjev2016}. Figure taken from \cite{Djidjev2016}. Copyright \textcopyright 2016 Society for Industrial and Applied Mathematics. Reprinted with permission. All rights reserved.\label{fig:charts}}
\end{figure}

\begin{figure}[t]
  \centering
  \includegraphics[width=0.9\textwidth]{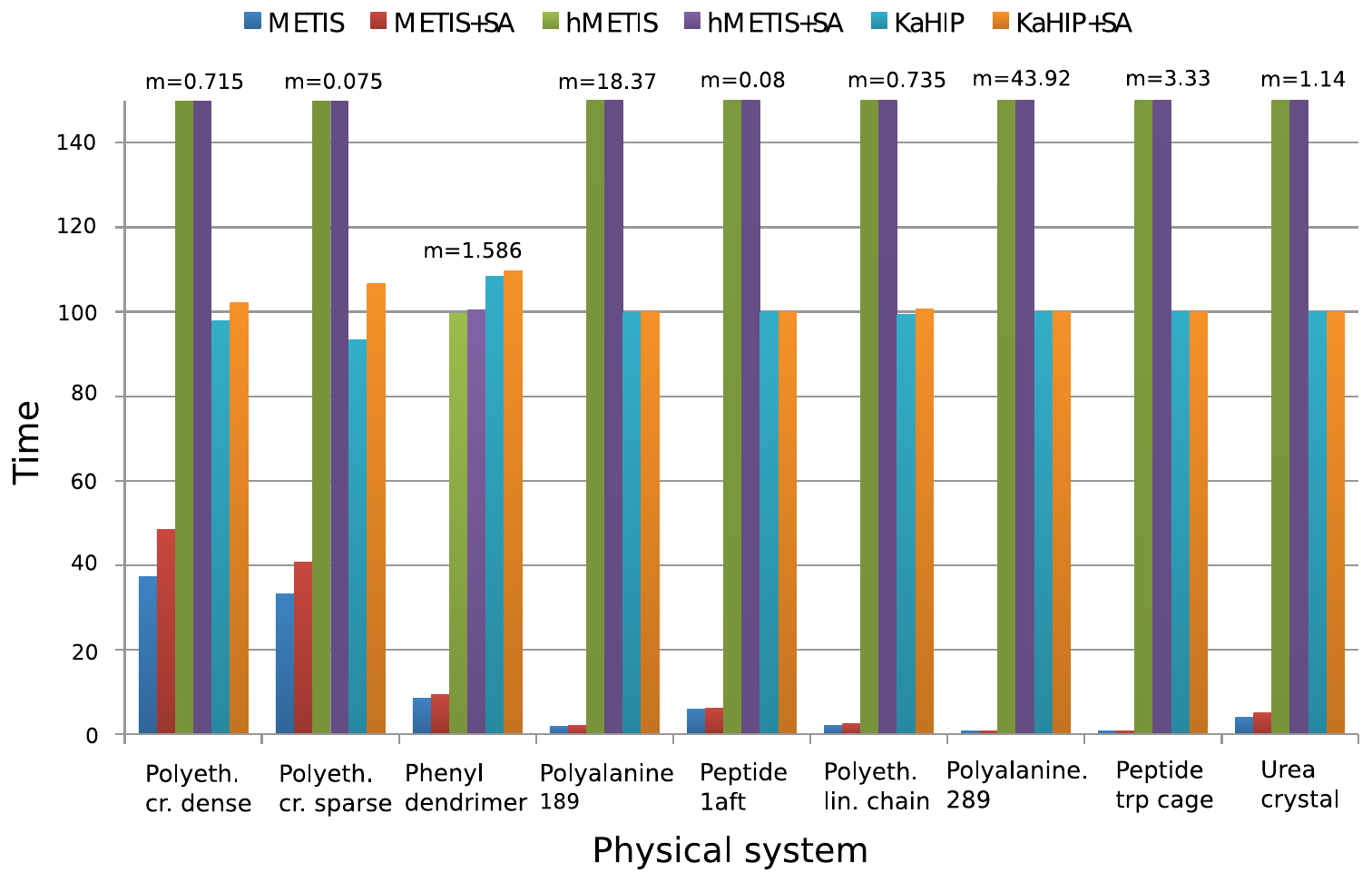}
  \caption{Computing time for partitioning. Test systems of Table~\ref{tab:systems}. We use the formatting of Figure~\ref{fig:charts} to handle the big discrepancy between values for different graphs. Figure taken from \cite{Djidjev2016}. Copyright \textcopyright 2016 Society for Industrial and Applied Mathematics. Reprinted with permission. All rights reserved.\label{fig:charts2}}
\end{figure}

Six methods are tested to partition each graph of Table~\ref{tab:systems} into 16 blocks:
\begin{enumerate}
  \item \textit{METIS} with parameters of Section~\ref{subsection_parameters};
  \item \textit{METIS} with subsequent simulated annealing (SA);
  \item \textit{hMETIS};
  \item \textit{hMETIS} with subsequent SA;
  \item \textit{KaHIP};
  \item \textit{KaHIP} with subsequent SA.
\end{enumerate}
As before, the sum of cubes \eqref{eq:sumOfCubes} criterion is used to assess the effectiveness of each method.

Figures~\ref{fig:charts} and~\ref{fig:charts2} show experimental results. We observe that all algorithms perform well (with the exception of the first two systems), and that \textit{METIS} and \textit{KaHIP} are significantly faster than \textit{hMETIS}. Importantly, post-processing with SA seems to improve solutions in almost all cases at negligible additional runtime, and is thus recommended.

Surprisingly, CH-partitions seem to pose a challenge for \textit{hMETIS} as its solutions are usually worse than those of the other two methods, and its runtime significantly exceeds the one of the other methods (an explanation of this remains for further reseach). We conclude that for these two reasons, \textit{hMETIS} seems unsuited for QMD simulations over longer time intervals, which is the aim of this work.

Although solutions returned by \textit{KaHIP} are of very good quality (measured with the sum of cubes criterion), the combination of \textit{METIS} and SA still outperforms \textit{KaHIP} while also having a shorter combined runtime.

The behavior of our algorithms seems to be dependent on the sparsity of the graph of a physical system. First, \textit{METIS} is able to outperform \textit{hMETIS} for denser graphs, but not for sparser ones. Second, the post-improvement of partitions with SA seems to be especially effective for dense graphs, which can be explained with the fact that dense graphs offer more possibilities to reassign and optimize edges than sparse graphs. This can be seen when applying the combination of \textit{METIS+SA} to the dense dendrimer system, see \cite[Table~$2$]{Djidjev2016} for details.

\begin{figure}
  \begin{center}
    \includegraphics[trim={2cm 2cm 2cm 2cm},clip,width=0.49\textwidth]{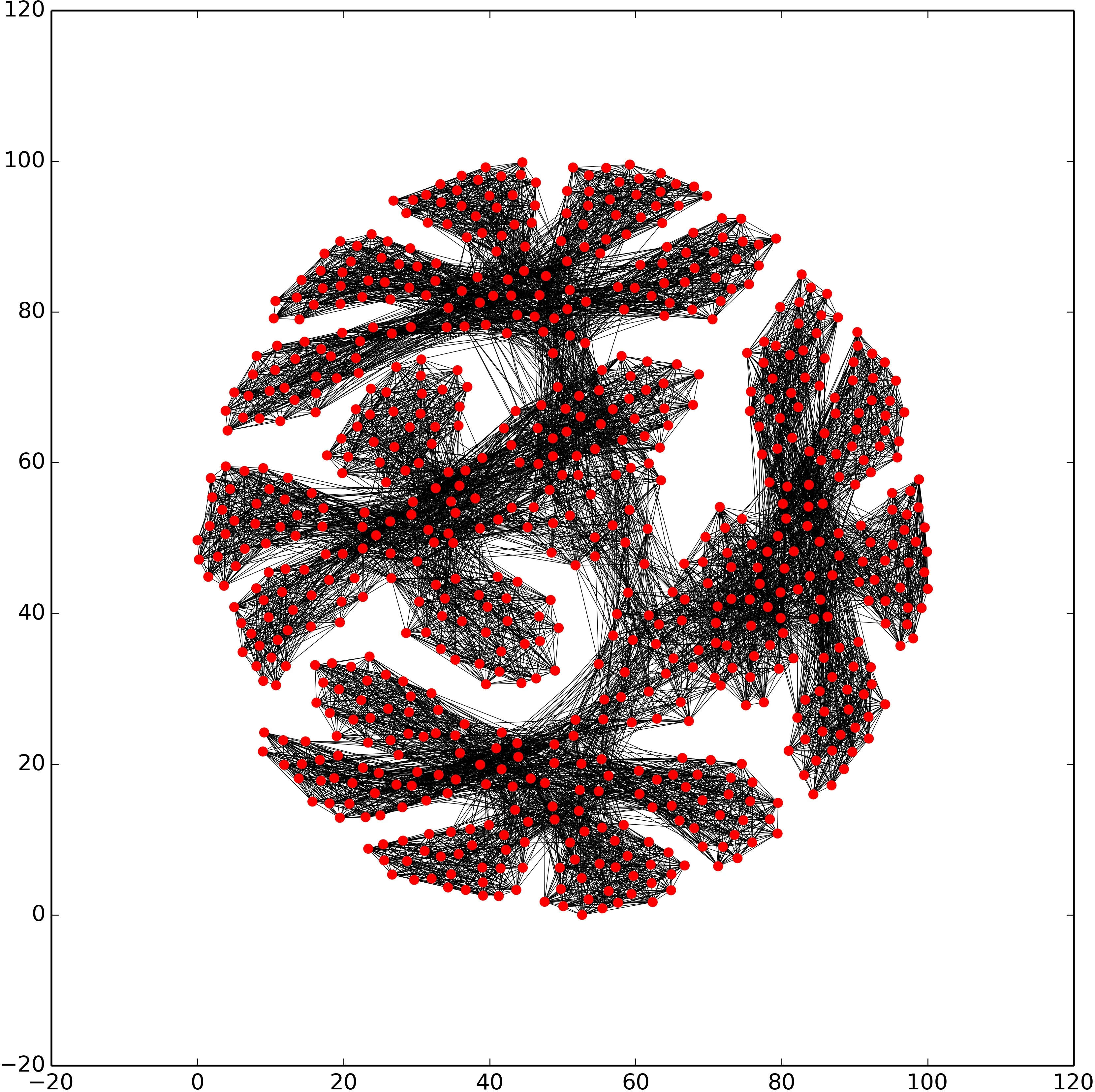}~
    \includegraphics[trim={2cm 2cm 2cm 2cm},clip,width=0.49\textwidth]{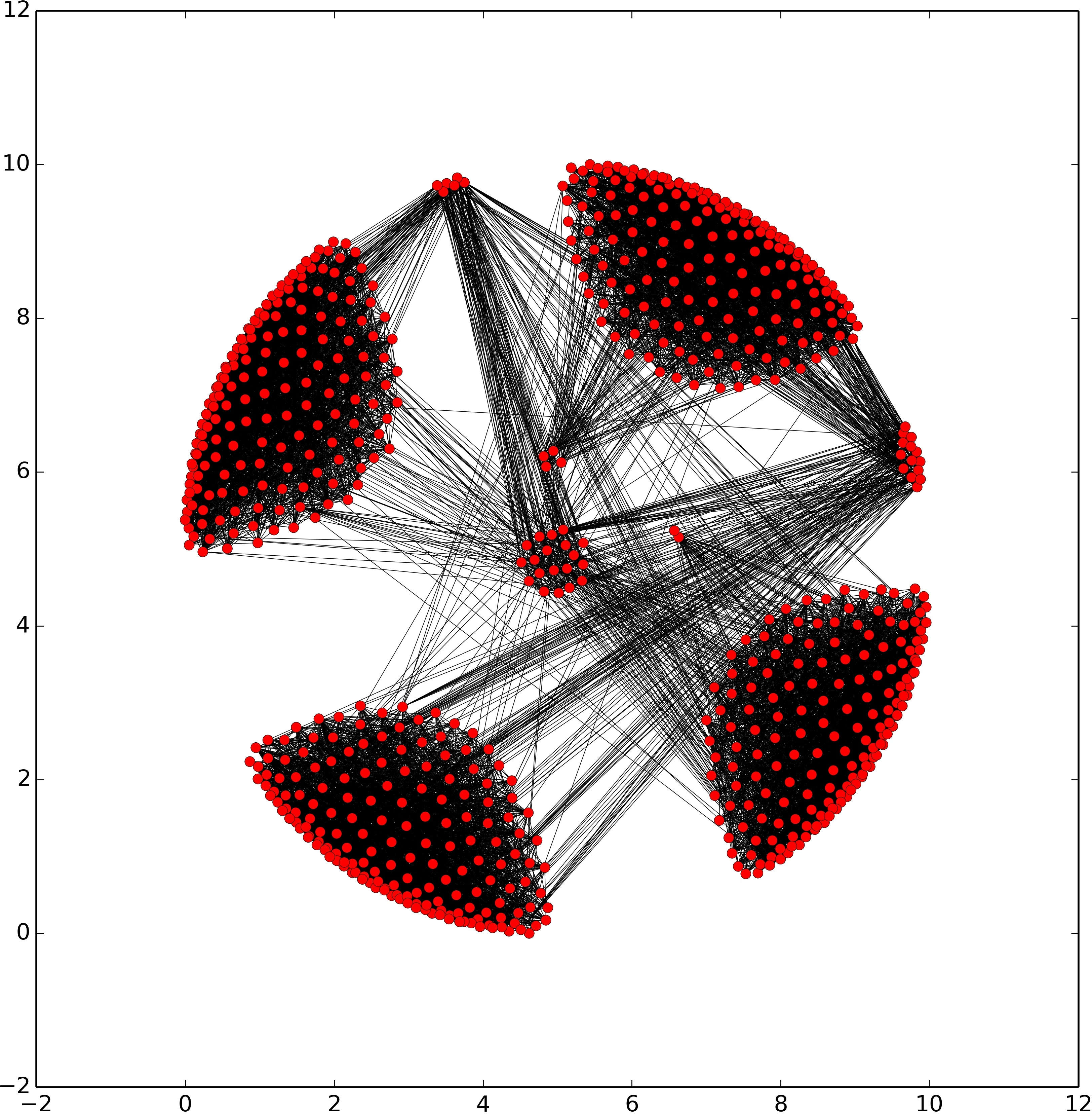}
    \caption{Left: Original graph extracted from the density matrix for the phenyl dendrimer molecular structure. Note the fractal-like structure of the graph. Right: Rearranged graph  by the partitions resulting from the METIS + SA algorithms. Only edges with weights larger than 0.01 were kept to ease visualization.\label{fig:globular}}
  \end{center}
\end{figure}

Figure~\ref{fig:globular} visualizes the relationship between the graph structure of a molecule (for the phenyl dendrimer molecule of Table~\ref{tab:systems}) and its graph partitioning obtained through \textit{METIS} and SA. After partitioning the molecular graph structure, the fractal-like structure of the phenyl dendrimer molecule and its dense components become clearly visible, as well as its sparse connections to other dense components. This structure is what our algorithm exploits to reduce the computations of the density matrix. Interestingly, SA sometimes dissolves entire blocks (meaning it produces blocks with no vertices), since such imbalanced partitionings still yield a further decrease of the objective function \eqref{eq:sumOfCubes}.

\subsection{Parallelized Implementation of G-SP2}
\label{subsection_scaling_protein}
We also assess the quality of the CH-partitions by measuring the speed-up when parallelizing the G-SP2 algorithm and applying it to real physical systems. In order to do this, the implementation of the G-SP2 algorithm of \citep{Niklasson2016} was modified to incorporate the graph partitioning step.

\begin{figure}
  \centering
  \includegraphics[width=0.8\textwidth]{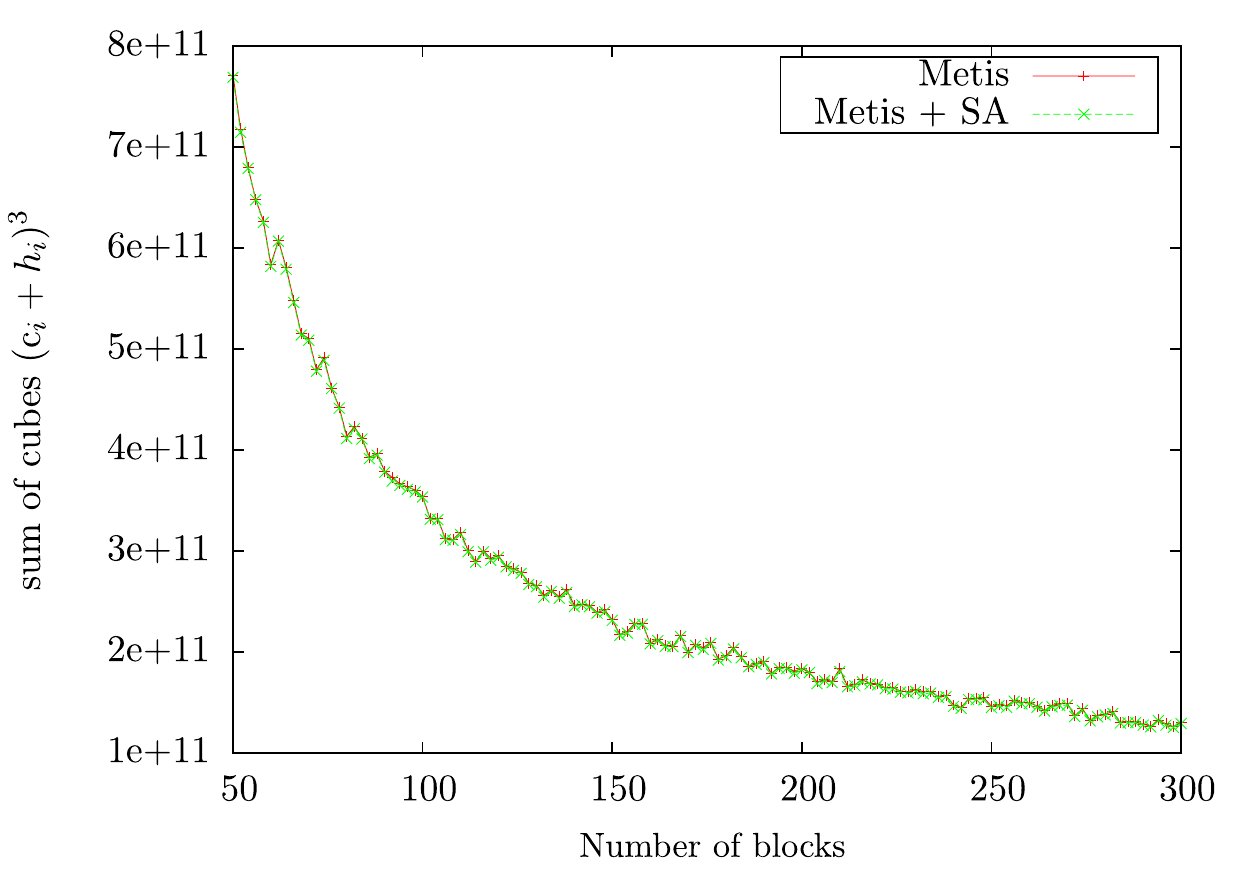}\\
  \includegraphics[width=0.8\textwidth]{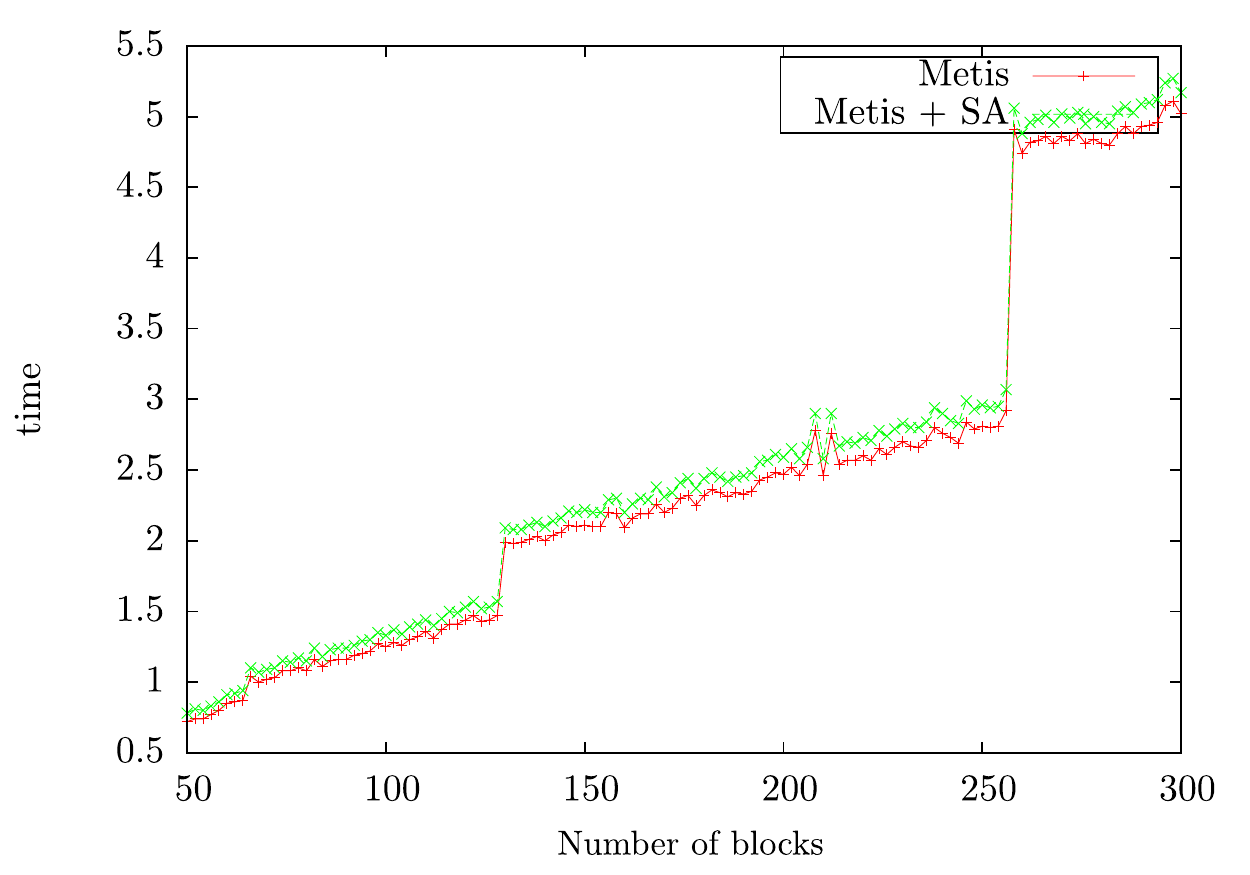} 
  \caption{Sum of cubes measure (top) and runtime to find the partitions (bottom) as a function of the number of blocks. Figure taken from \cite{Djidjev2016}. Copyright \textcopyright 2016 Society for Industrial and Applied Mathematics. Reprinted with permission. All rights reserved.\label{fig:protein_scale}}
\end{figure}

In particular, to measure the speed-up, we run \textit{METIS} and \textit{METIS+SA} on the \textit{polyalanine 259} protein system of Table~\ref{tab:systems} and record the obtained CH-partitions. Computations were carried out with the Wolf IC cluster of Los Alamos National Laboratory, whose computing nodes have 2 sockets each containing an 8-core Intel Xeon SandyBridge E5-2670, amounting to a total of 16 cores per computing node. A total of 32 GB of RAM are shared between the sixteen cores on each node, which are connected using a Qlogic Infiniband (IB) Quad Data Rate (QDR) network in a fat tree topology using a 7300 Series switch. Parallelization across nodes was done with OpenMPI (version 1.8.5), and parallelization across cores within a node was done with OpenMP~4.

The sum of cubes measure and the computing time are displayed in Figure~\ref{fig:protein_scale} as a function of the number of CH-partitions for the \textit{polyalanine 259} protein system. We observe in Figure~\ref{fig:protein_scale} (top) that the total effort of G-SP2, measured with the sum of cubes criterion, decreases monotonically as a function of the number of blocks and parallelized subproblems.

Figure~\ref{fig:protein_scale} (bottom) displays the computing time for the graph partitioning step alone. We observe that the partitioning effort increases with the number of blocks. The steps occurring at $65$, $129$, $257$ etc.\ partitions in the plot can easily be explained: Every time the number of blocks surpasses a power of two, the multilevel algorithm which \textit{METIS} is based on bisects the partitioning problem into one more (recursive) layer. Figure~\ref{fig:protein_scale} (bottom) also displays that employing the SA post-processing step only adds a minimal additional effort to the overall computation (when compared to the graph partitioning step alone). The usage of SA thus seems very sensible in light of the improvements it achieves when applied to the edge cut optimized partitions computed by a conventional algorithm. To summarize, Figure~\ref{fig:protein_scale} demonstrates that the total effort of the G-SP2 algorithm decreases when applied in parallel despite the increasing effort to compute a partitioning.

\subsection{Single Node SM-SP2 versus Parallelized Implementation of G-SP2}
\label{subsection_scaling_prot}

\begin{figure}
  \centering
  \includegraphics[width=0.7\textwidth]{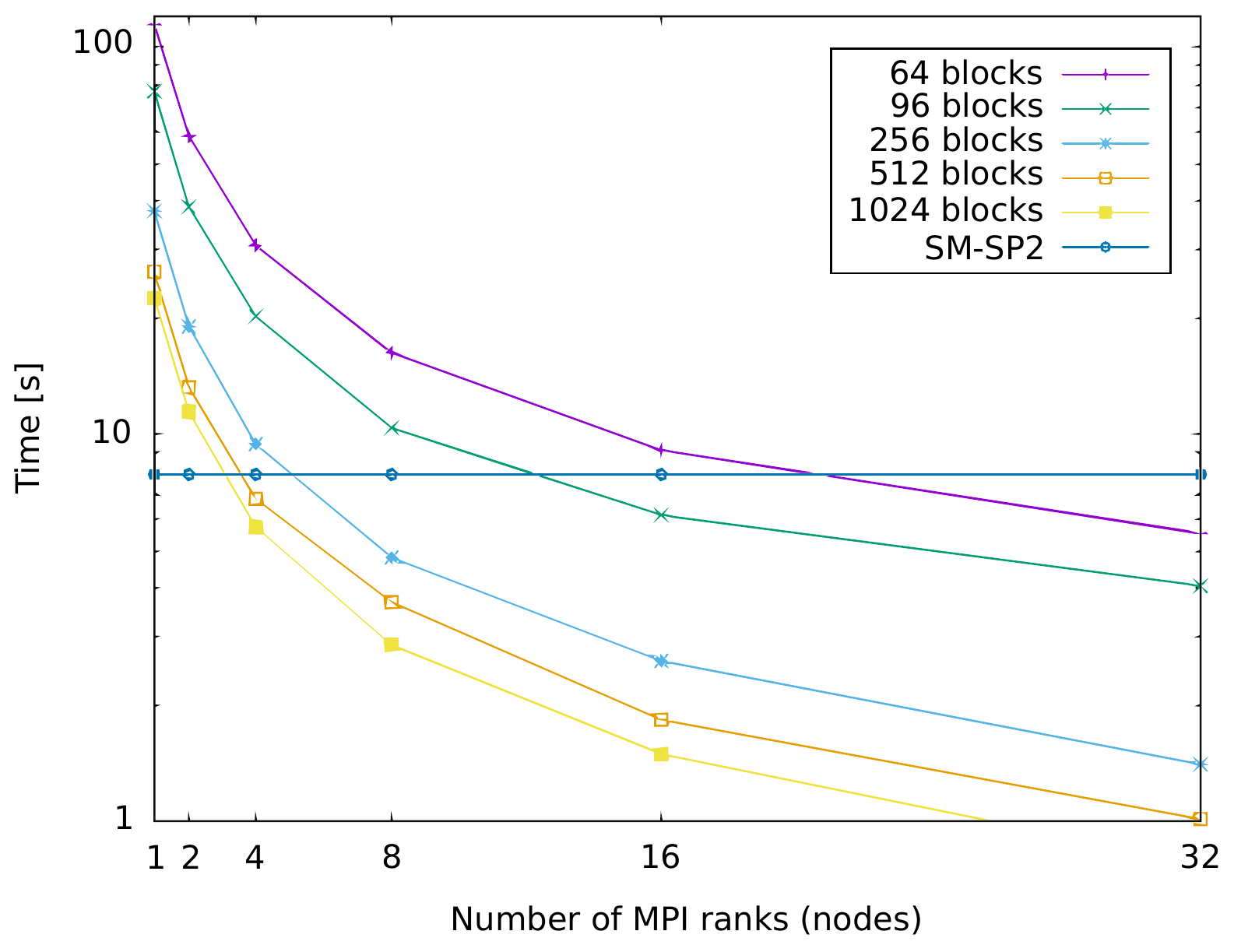}
  \caption{Runtime of parallelized G-SP2 as a function of the number of nodes. Different numbers of blocks. G-SP2 applied to the polyalanine 259 molecule. Figure taken from \cite{Djidjev2016}. Copyright \textcopyright 2016 Society for Industrial and Applied Mathematics. Reprinted with permission. All rights reserved.\label{fig:SueChristian}}
\end{figure}
Figure~\ref{fig:SueChristian} aims to quantify the computational savings over single node G-SP2 when running our proposed parallel SP2 algorithm. For this, Figure~\ref{fig:SueChristian} compares the runtime for our parallel G-SP2 on $1$-$32$ nodes against a threaded single node implementation of SM-SP2. We use G-SP2 here since the communication overhead exceeds the gain obtained by the extra computing power in a multi-node implementation, the main motivation for developing G-SP2 in the first place. As before, we employed \textit{METIS} with parameters specified in Section~\ref{subsection_parameters} together with SA for post-processing. The test system is again the polyalanine 259 molecule.

Figure~\ref{fig:SueChristian} shows that, as expected, both an increasing number of nodes and an increasing number of blocks decreases the G-SP2 runtime. When only few nodes are used for parallelization, the decrease in runtime is most pronounced since then, increasing the number of parallel nodes causes the runtime to drop sharply. For a higher number of nodes the curves somewhat flatten out.

Due to the overhead from the parallel G-SP2 computation, the parallelized run of SP2 is actually slower than the SM-SP2 computation on a single node for low numbers of nodes (between $4$ and $16$ nodes depending on the number of blocks). The runtime decreases with an increase in the number of nodes. Eventually, the effort falls below the one of a single node implementation. For the particular physical system of the polyalanine 259 molecule, a computational speed-up is only observed for at least $4$ nodes.

\subsection{Relationship between Molecular System and Partitions}
\label{subsection_partitioning_relationship}
\begin{figure}[t]
  \begin{center}
    \includegraphics[width=0.8\textwidth]{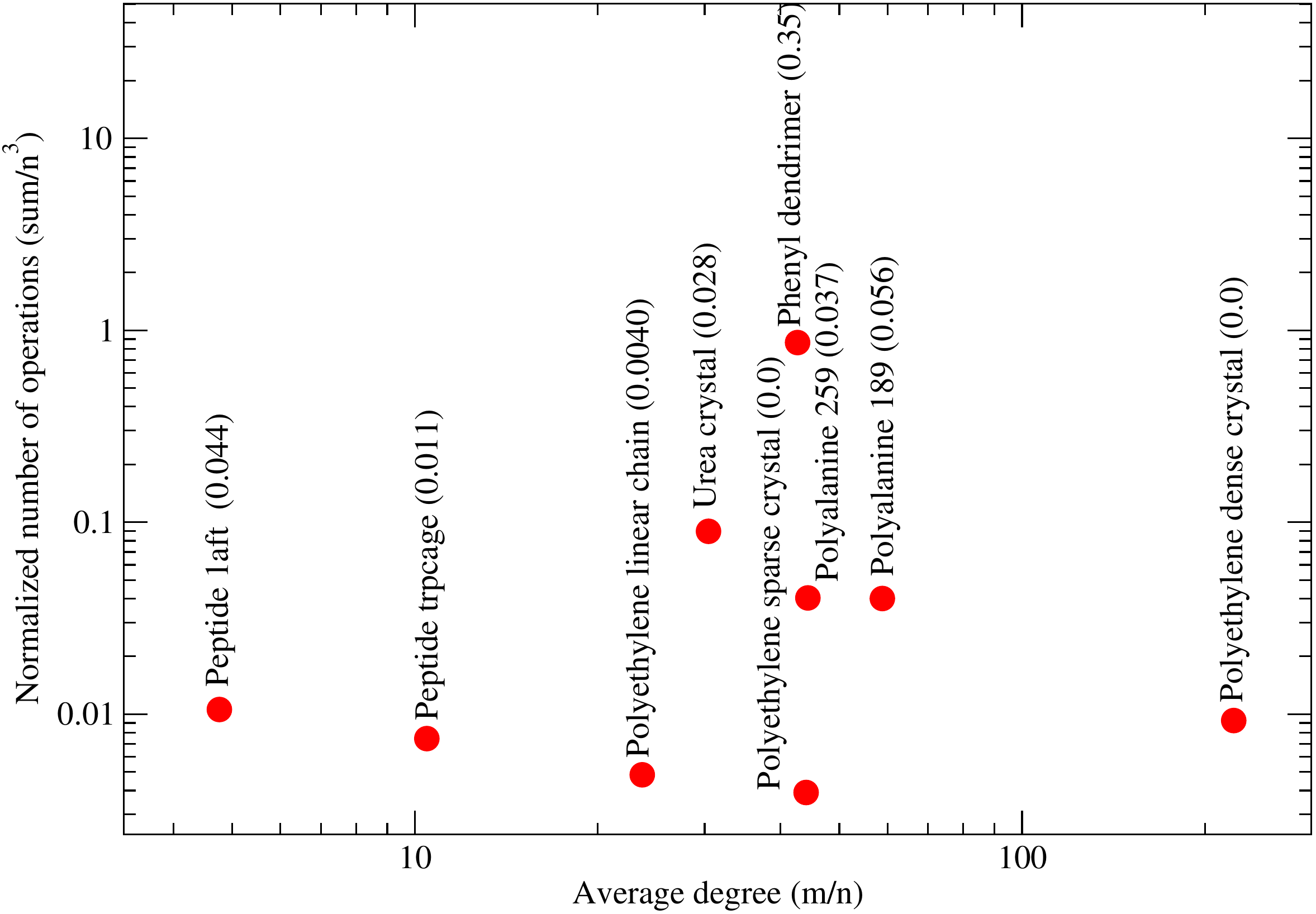}
    \caption{Sum of cubes criterion \eqref{eq:sumOfCubes} normalized by the complexity $n^3$ of the dense matrix-matrix multiplication for G-SP2 with METIS. Number of vertices $n$, number of edges $m$, and average degree $m/n$. Test systems of Table~\ref{tab:systems}. Brackets show fractions of $(\max-\min)/n$ (MMPN). Figure taken from \cite{Djidjev2016}. Copyright \textcopyright 2016 Society for Industrial and Applied Mathematics. Reprinted with permission. All rights reserved.\label{fig:matrices}}
  \end{center}
\end{figure}

Figure~\ref{fig:matrices} does not seem to exhibit a correlation between the molecular connectivity (measured with the average graph degree) and the normalized number of operations (NNO), defined as the sum of cubes criterion \eqref{eq:sumOfCubes} normalized by the complexity $n^3$ of the dense matrix-matrix multiplication. The observation also applies to the polyethylene dense crystal, whose normalized number of operations remains low although its average degree is high. Similar observations can be made for other molecules with a smaller average degree.

Our G-SP2 algorithm with METIS being used in the graph partitioning step is capable of finding the lowest NNO for 1-dimensional systems such as the polyethylene linear chain and the polyethylene sparse crystal (compare the MMPN values in the brakets in Figure~\ref{fig:matrices} to the measured normalized number of operations on the y-axis). According to \cite{TF9393500482}, regular agglomerates of polyethylene chains align along a particular direction with a large chain-to-chain distance.

We do not observe any advantage of our approach (measured via NNO) for \textit{regular} systems (polyethylene linear chain, polyethylene sparse crystal, polyethylene dense crystal and urea crystal).

A difficult case for our graph partitioning task is the phenyl dendrimer (expressed through its high NNO values). This is due to the fractal-like structure of the graph associated with the molecule.

Another class with a large NNO are proteins (solvated polyalanines). We conjecture that the large average node degree (in comparison to peptides, i.e.\ small proteins) is responsible for the large NNO measurements. Based on Figure~\ref{fig:matrices}, we further conjecture that the difference of maximum and minimum partition norms $(\max-\min)/n$ (MMPN) is correlated with the NNO. Here, the number of vertices is denoted as $n$, and the maximal and minimal sizes of the obtained blocks are denoted as $\max$ and $\min$, respectively. The conjectured correlation is clearly visible in Figure~\ref{fig:matrices}: the dendrimer tends to both large MMPN and NNO values, proteins exhibit intermediate values of both MMPN and NNO and finally, we observe low values of both MMPN and NNO for sparse ordered systems such as polyethyene chains.

\section{Discussion}
\label{section_discussion}
This paper speeds up the computation of the density matrix in MD simulations through parallelization, informed by graph partitioning applied to the structure graph underlying a molecule. Our experimental results are based on graphs derived from density matrices of physical systems.

In our article we focus on a certain flavor of the classical graph partitioning problem arising from molecular dynamics simulations. In contrast to classical edge cut partitioning, we minimize blocks with respect to both the number of their core vertices and the number of their neighbors in adjacent blocks (halos). To the best of our knowledge, this type of graph partitioning (which we coin \textit{CH-(core-halo)-partitioning}) has not been studied previously.

This work makes two contributions. First, the CH-partitioning problem under consideration is mathematically described and justified. We prove that the partitioned evaluation of a matrix polynomial is equivalent to the evaluation of the original (unpartitioned) matrix given a sufficient condition is satisfied. Second, we evaluate several approaches for computing CH-partitions using three performance measures: the total computational effort, the maximal effort per processor, and the overall computational runtime. Special focus is given to the post-processing of partitions obtained with conventional graph partitioning algorithms for which we use our own modified SA approach.

We find that our flavor of the partitioning problem can be solved using standard graph partitioning packages. Moreover, post-optimization of the partitions obtained through classical graph partitioning packages can be performed well with our SA scheme. As expected, the time to evaluate matrix polynomials for different system (graph) sizes decreases with both the number of processors and blocks in our simulations. Our main result is that the increased effort for graph partitioning and post-optimization with SA is beneficial overall when applying our parallelized version of the G-SP2 algorithm to meaningful physical systems. Based on our observation that \textit{METIS} with a SA post-processing step is significantly faster than competing methods while giving the best results on average, we recommend this combination for practical use. Of course, the research in this article can be always updated by including newer partitioning algorithms as time progresses, such as KaHyPar~\citep{KaHyPar}. Additionally, a consideration of the statistical significance of our results could be of interest as future work.

\section*{Acknowledgments}
The authors would like to thank Vivek B.\ Sardeshmukh for his help with various aspects related to graph partitioning, Purnima Ghale and Matthew Kroonblawd for their help with selecting meaningful physical datasets of real-world molecules, and Ben Bergen, Nick Bock, Marc Cawkwell, Christoph Junghans, Robert Pavel, Sergio Pino, Jerry Shi, and Ping Yang for their feedback.

\bibliographystyle{apalike}

\begin{thebibliography}{}
\bibitem[Bader et~al., 2013]{DBLP:conf/dimacs/2012}
Bader, D.~A., Meyerhenke, H., Sanders, P., and Wagner, D., editors (2013).
\newblock {\em Graph Partitioning and Graph Clustering - 10th {DIMACS}
  Implementation Challenge Workshop}, volume 588 of {\em Contemporary
  Mathematics}. AMS.

\bibitem[Bock and Challacombe, 2013]{Bock}
Bock, N. and Challacombe, M. (2013).
\newblock An optimized sparse approximate matrix multiply for matrices with
  decay.
\newblock {\em SIAM Journal on Scientific Computing}, 35(1):C72--C98.

\bibitem[Borstnik et~al., 2014]{Hutter}
Borstnik, U., VandeVondele, J., Weber, V., and Hutter, J. (2014).
\newblock Sparse matrix multiplication: The distributed block-compressed sparse
  row library.
\newblock {\em Parallel Computing}, 40:47--58.

\bibitem[Bunn, 1939]{TF9393500482}
Bunn, C. (1939).
\newblock {The crystal structure of long-chain normal paraffin hydrocarbons.
  The "shape" of the CH2 group.}
\newblock {\em Trans. Faraday Soc.}, 35:482--491.

\bibitem[Djidjev et~al., 2016]{Djidjev2016}
Djidjev, H.~N., Hahn, G., Mniszewski, S.~M., Negre, C.~F., Niklasson, A.~M.,
  and Sardeshmukh, V. (2016).
\newblock Graph partitioning methods for fast parallel quantum molecular
  dynamics (full text with appendix).
\newblock {\em SIAM Workshop on Combinatorial Scientific Computing (CSC16)}.

\bibitem[Elstner et~al., 1998]{Elstner1998}
Elstner, M., Porezag, D., Jungnickel, G., Elsner, J., Haugk, M., Frauenheim,
  T., Suhai, S., and Seifert, G. (1998).
\newblock Self-consistent-charge density-functional tight-binding method for
  simulations of complex materials properties.
\newblock {\em Phys. Rev. B}, 58(11):7260--7268.

\bibitem[Fiduccia and Mattheyses, 1982]{FiducciaMattheyses1982}
Fiduccia, C. and Mattheyses, R. (1982).
\newblock A linear time heuristic for improving network partitions.
\newblock {\em In Proc. 19th IEEE Design Automation Conference}, pages
  175--181.

\bibitem[Finnis et~al., 1998]{MFinnis98}
Finnis, M.~W., Paxton, A.~T., Methfessel, M., and van Schilfgarde, M. (1998).
\newblock Crystal structures of zirconia from first principles and
  self-consistent tight binding.
\newblock {\em Phys. Rev. Lett.}, 81:5149.

\bibitem[Ford~Jr. and Fulkerson, 1956]{FordFulkerson1956}
Ford~Jr., L. and Fulkerson, D. (1956).
\newblock Maximal flow through a network.
\newblock {\em Canad. J. Math.}, 8:399--404.

\bibitem[Frauenheim et~al., 2000]{TFrauenheim00}
Frauenheim, T., Seifert, G., aand Z.~Hajnal, M.~E., Jungnickel, G., Poresag,
  D., Suhai, S., and Scholz, R. (2000).
\newblock A self-consistent charge density-functional based tight-binding
  method for predictive materials simulations in physics, chemistry and
  biology.
\newblock {\em Phys. Stat. sol.}, 217:41.

\bibitem[Karypis and Kumar, 1998]{hmetis_manual}
Karypis, G. and Kumar, V. (1998).
\newblock A hypergraph partitioning package.

\bibitem[Karypis and Kumar, 1999]{KarypisKumar1999}
Karypis, G. and Kumar, V. (1999).
\newblock A fast and high quality multilevel scheme for partitioning irregular
  graphs.
\newblock {\em SIAM J. Sci. Comput.}, 20(1):359--392.

\bibitem[Karypis and Kumar, 2000]{KarypisKumar2000}
Karypis, G. and Kumar, V. (2000).
\newblock Multilevel k-way hypergraph partitioning.
\newblock {\em VLSI Design}, 11(3):285--300.

\bibitem[Kirkpatrick et~al., 1983]{Kirkpatrick1983}
Kirkpatrick, S., Gelatt~Jr, C., and Vecchi, M. (1983).
\newblock Optimization by simulated annealing.
\newblock {\em Science}, 200(4598):671--680.

\bibitem[Mniszewski et~al., 2015]{Mniszewski2015}
Mniszewski, S.~M., Cawkwell, M.~J., Wall, M., Moyd-Yusof, J., Bock, N.,
  Germann, T., and Niklasson, A.~M. (2015).
\newblock Efficient parallel linear scaling construction of the density matrix
  for born-oppenheimer molecular dynamics.
\newblock {\em J Chem Theory Comput}, 11(10):4644--4654.

\bibitem[Niklasson, 2002]{Niklasson2002}
Niklasson, A.~M. (2002).
\newblock Expansion algorithm for the density matrix.
\newblock {\em Phys. Rev. B}, 66(15):155115--155121.

\bibitem[Niklasson et~al., 2016]{Niklasson2016}
Niklasson, A.~M., Mniszewski, S.~M., Negre, C.~F., Cawkwell, M.~J., Swart,
  P.~J., Mohd-Yusof, J., Germann, T.~C., Wall, M.~E., Bock, N., Rubensson,
  E.~H., and Djidjev, H.~N. (2016).
\newblock Graph-based linear scaling electronic structure theory.
\newblock {\em J Chem Phys}, 144(23):234101.

\bibitem[P{\i}nar and Hendrickson, 2001]{PinarHendrickson2001}
P{\i}nar, A. and Hendrickson, B. (2001).
\newblock Partitioning for complex objectives.
\newblock {\em Proceedings of the 15th International Parallel and Distributed
  Processing Symposium (CDROM), IEEE Computer Society, Washington, DC, USA},
  pages 1--6.

\bibitem[Sanders and Schulz, 2011]{SandersSchulz2011}
Sanders, P. and Schulz, C. (2011).
\newblock {Engineering multilevel graph partitioning algorithms}.
\newblock In {\em LNCS}, volume 6942, pages 469--480.

\bibitem[Sanders and Schulz, 2013]{sandersschulz2013}
Sanders, P. and Schulz, C. (2013).
\newblock {Think Locally, Act Globally: Highly Balanced Graph Partitioning}.
\newblock In {\em International Symposium on Experimental Algorithms (SEA)},
  volume 7933 of {\em LNCS}, pages 164--175. Springer.

\bibitem[Schlag et~al., 2015]{KaHyPar}
Schlag, S., Henne, V., Heuer, T., Meyerhenke, H., Sanders, P., and Schulz, C.
  (2015).
\newblock k-way hypergraph partitioning via n-level recursive bisection.
\newblock {\em arXiv:1511.03137}, pages 1--21.

\bibitem[VandeVondele et~al., 2012]{VandeVondele}
VandeVondele, J., Borštnik, U., and Hutter, J. (2012).
\newblock Linear scaling self-consistent field calculations with millions of
  atoms in the condensed phase.
\newblock {\em Journal of Chemical Theory and Computation}, 8(10):3565--3573.

\bibitem[von Looz et~al., 2016]{Looz2016}
von Looz, M., Wolter, M., Jacob, C.~R., and Meyerhenke, H. (2016).
\newblock Better partitions of protein graphs for subsystem quantum chemistry.
\newblock {\em arXiv:1606.03427}, pages 1--20.
\end{thebibliography}

\appendix
\section{Proofs for Section \ref{section_theory}}
\label{appendix_proofs}
For any graph $I$ and vertex $v \in I$, the \textit{neighborhood} of $v$ in $I$ is the set $N(v,I)=\{w\in V(I)~|~(v,w)\in E(I)\}$. Let $H=\PP(G)$, $v$ be a vertex of $G$, and $H_{v}$ denote the subgraph of $H$ induced by $N(v,H)$. For the following lemmas we assume that $T_i\cap E(H)=\emptyset$ for all $i$, i.e., none of the edges in $H=\PP(G)$ are thresholded. (With thresholding, the partitioned matrix algorithm is not exact, but approximate.) We have the following properties.

\begin{lemma}
  \label{lem:neighb}
  Let $v$ be a vertex of $G$. Then $N(v,\PP(G))=N(v,\PP(H_v))$.
\end{lemma}

\begin{proof}
First we prove that $N(v,\PP(G))\subseteq N(v,\PP(H_v))$. Let $w\in N(v,\PP(G))$. Then $(v,w)\in E(\PP(G))$ and hence $(v,w)\in E(H_v)\subseteq E(H)$. Since by assumption $T_i\cap E(H)=\emptyset$, $(v,w)\not\in T_i$ for all $i$. From $(v,w)\in E(H_v)$, the last relation, and the fact that all vertices of $H$ have loops, $(v,w)\in E(\PP(H_v))$. Hence, $w\in N(v,\PP(H_v))$.

Now we prove that $N(v,\PP(H_v))\subseteq N(v,\PP(G))$. Let $w\in N(v,\PP(H_v))$. Since $\PP(H_v)$ and $H_v$ have the same vertex sets, we have $w\in N(v,H_v)$. Furthermore, since $H_v$ is a subgraph of $H$, $w\in N(v,H)=N(v,\PP(G))$.
\end{proof}

The lemma shows that $v$ has the same neighbors in $\PP(G)$ and $\PP(H_v)$, i.e., their corresponding matrices have nonzero entries in the same positions in the row (or column) corresponding to $v$. We will next strengthen that claim by showing that the corresponding nonzero entries contain equal values.

Let $X_{v}$ be the submatrix of $A$ defined by all rows and columns that correspond to vertices of $V(H_{v})$. We will call vertex $v$ the \emph{core} and the remaining vertices the \textit{halo} of $V(H_{v})$. We define the set $\{V(H_{v})~|~v\in G\}$ to be the \emph{CH-partition} of $G$. Note that, unlike other definitions of a partition used elsewhere, the vertex sets of CH-partitions (and, specifically, the halos) can be, and typically are, overlapping.

\begin{lemma}
  \label{lem:matrixValues}
  For any $v\in V(G)$ and any $w\in N(v,\PP(G))$, the element of $P(A)$ corresponding to the edge $(v,w)$ of $\PP(G)$ is equal to the element of $P(X_v)$ corresponding to the edge $(v,w)$ of $\PP(H_v)$.
\end{lemma}
\begin{proof}
Let $2^s$ be the degree of $P$. We will prove the lemma by induction on $s$. Clearly, the claim is true for $s=0$ since the elements of both $A^1$ and $X^1$ are original elements of the matrix $A$. Assume the claim is true for $s-1$. Define $P'=P_1\mathsmaller{\circ} T_1\mathsmaller{\circ}\dots\mathsmaller{\circ} P_{s-1}\mathsmaller{\circ} T_{s-1}$. By the inductive assumption, the corresponding elements in the matrices $A'=P'(A)$ and $X'=P'(X)$ have equal values. We need to prove the same for the elements of $A'^2$ and $X'^{2}$.

Let $(v,w)\in E(\PP(G))$. By Lemma~\ref{lem:neighb}, $(v,w)\in E(\PP(H_v))$. For each vertex $u$ of $\PP(H_v)$ let $u'$ denote the corresponding row/column of $X$. We want to show that $P(A)(v,w)=P(X)(v',w')$.

By definition of matrix product, $A'^2(v,w)=\sum A'(v,u) A'(u,w)$, where the summation is over all $u$ such that $(v,u),(u,w)\in E(\PP(G))$. Similarly, $X'^2(v',w')=\sum X'(v',u')X'(u',w')$, where the summation is over the values of $u'$ corresponding to the values of $u$ from the previous formula, by Lemma~\ref{lem:neighb}. By the inductive assumption, $A'(v,u)=X'(v',u')$ and $A'(u,w)=X'(u',w')$, thus $A'^2(v,w)=X'^2(v',w')$. Since by assumption $A'(v,w)=X'(v',w')$, we have $P_s(A')(v,w)=P_s(X')(v',w')$, and hence $P(A)(v,w)=(P_s\mathsmaller{\circ}T_s)(A')(v,w)=(P_s\mathsmaller{\circ}T_s)(X')(v',w')=P(X)(v',w')$.
\end{proof}

Clearly, in many cases it will be advantageous to consider CH-partitions whose cores contain multiple vertices. We will next show that the above approach for CH-partitions with single-node cores can be generalized to the multi-node core case.

We will generalize the definitions of $N(v,\PP(G))$ and $N(v,\PP(H_v))$ for the case where the vertex $v$ is replaced by a set $U$ of vertices of $G$. For any graph $I$, we define $N(U,I)=\bigcup_{v\in U}N(v,I)$. Furthermore, we define by $H_U$ the subgraph of $H$ induced by $N(U,H)$.

Suppose the sets $\{U_i~|~i=1,\dots,q\}$ are such that $\bigcup_{i}U_i=V(G(A))$ and $U_i\cap U_j=\emptyset$. In this case we can define a CH-partition of $G=G(A)$ consisting of $q$ sets, where for each $i$, $U_i$ is the core and $N(U_i,H)\setminus U_i$ is the halo of $\Pi_i$.

The following generalizations of Lemma~\ref{lem:neighb} and Lemma~\ref{lem:matrixValues} follow in a straightforward manner.

\begin{lemma}
  \label{lem:neighbSet}
  Denote by $H_i$ the subgraph $\PP(H_{U_i})$ of $G$. Let $v$ be a vertex of $U_i$. Then $N(v,\PP(G))=N(v,H_i)$.
\end{lemma}

Denote by $A_{U_i}$ the submatrix of $A$ consisting of all rows and columns that correspond to vertices of $V(H_{U_i})$. The following main result of this section shows that $P(A)$ can be computed on submatrices of the Hamiltonian: this insight justifies the parallelized evaluation of a matrix polynomial.

\begin{lemma}
  \label{lem:matrixValuesMulti}
  For any $v\in U_i$ and any $w\in N(v,\PP(G))$, the element of $P(A)$ corresponding to the edge $(v,w)$ of $\PP(G)$ is equal to the element of $P(A_{U_i})$ corresponding to the edge $(v,w)$ of $H_i$.
\end{lemma}

In case of a QMD simulation, as outlined in Section~\ref{section_theory}, we assume that the sparsity structure of $P(A)$ can be approximated by the one of the density matrix $D$ from the previous QMD simulation step. For the halos we approximate $H=\PP(G)$ with $G(D)$ as before, and use the current $H$ as a substitute for the halos. This leads to the generalization of the algorithm for computing $P(A)$ in case of multi-vertex cores given at the end of Section~\ref{section_theory}.

\section{Further Details on Experimental Results}
\label{app:experiments}
Raw data for the experiments described in Section~\ref{subsection_assessment_all} can be found in Table~\ref{tab:results}.
The six partitioning schemes we use are listed in column ``methods''. The performance of those methods is measured in four different ways: As a measure of the total matrix multiplication cost of a step of the SP2 algorithm, we report the sum of cubes criterion (\ref{eq:sumOfCubes}) in column $3$ (``sum''). The sum of cubes criterion is also a measure of the computational effort of SP2 since matrix multiplications consume most of the computation time in SP2. As a measure of the variability of block sizes created by our algorithm, we report the smallest and largest block size of any CH-partition (columns $4$ (``min'') and $5$ (``max'')). In the best case scenario, the sizes of all blocks should be roughly equal since otherwise, the nodes or processors will have very unequal computational loads in our parallel implementation of the SP2 algorithm. This is undesirable in practice. The average computation time (in seconds) for each partitioning algorithm is shown in the last column (``time''). As pointed out by a referee, it would be interesting to investigate closer how the discrepancy of the block sizes from the mean influences the performance of SP2: this is left as future research.

\begin{table}
  \caption{Various test systems (first column) evaluated by different partition schemes (second column; number of vertices $n$, edges $m$, and blocks $p$ are given). Results measured with the sum of cubes (\ref{eq:sumOfCubes}) criterion (sum), the sum of size and halo for the smallest CH-partition (min) as well as the biggest CH-partition (max), and the overall runtime (in seconds) are given in columns 3-6.\vspace{-2mm}\label{tab:results}}
  \begin{center}
	\scriptsize
    \begin{tabular}{|l||l|r|r|r|r|}
		\hline
		Test system & method & sum & min & max & time [s]\\
		\hline
		polyethylene dense crystal &METIS & 57,982,058,496 & 1536 & 1536 &  0.267175 \\ 
		n = 18432 & METIS + SA & 51,856,752,364 & 976 & 1536 & 0.347209 \\ 
		m = 4112189 & HMETIS & 7,126,357,377,024 & 3840 & 9984 & 141.426 \\ 
		p = 16 & HMETIS + SA & 1,362,943,612,944 & 2520 & 5814 & 141.79 \\ 
		& KaHIP & 32,614,907,904    &     768    &    1536   &  0.7 \\ 
		& KaHIP + SA &  32,614,907,904  &  768   & 1536  &  0.73 \\
		\hline

		polyethylene sparse crystal &METIS & 24,461,180,928 & 1152 & 1152 &  0.024942 \\ 
		n = 18432 & METIS + SA & 24,461,180,928 & 1152 & 1152 & 0.030508 \\ 
		m = 812343 & HMETIS & 195,689,447,424 & 2304 & 2304 & 55.9726 \\ 
		p = 16 & HMETIS + SA & 170,056,587,295 & 2013 & 2299 & 55.9943 \\ 
		& KaHIP & 24,461,180,928     &   1152    &    1152 &   0.07\\ 
		& KaHIP + SA &  24,461,180,928   &     1152    &    1152 &   0.08 \\
		\hline

		phenyl dendrimer &METIS & 336,049,081 & 150 & 409 &  0.13482 \\ 
		n = 730 & METIS + SA & 146,550,740 & 0 & 382 & 0.14877 \\ 
		m = 31147 & HMETIS & 177,436,462 & 135 & 358 & 1.578 \\ 
		p = 16 & HMETIS + SA & 118,409,940 & 0 & 358 & 1.59436 \\ 
		& KaHIP & 231,550,645        &  55      &   381  &  1.72\\ 
		& KaHIP + SA &  116,248,715    &       0      &   324 &   1.74 \\
		\hline

		polyalanine 189 &METIS & 1,305,573,505,507 & 3358 & 5145 &  0.332091 \\ 
		n = 31941 & METIS + SA & 1,297,206,329,828 & 3362 & 5093 & 0.372463 \\ 
		m = 1879751 & HMETIS & 1,402,737,703,273 & 3762 & 5124 & 418.229 \\ 
		p = 16 & HMETIS + SA & 1,393,115,476,879 & 3765 & 5119 & 418.28 \\ 
		& KaHIP &  1,649,301,823,304       &   12 &       6030  & 18.35 \\ 
		& KaHIP + SA & 1,624,800,725,049     &     12    &    5983 &  18.39  \\
		\hline

		peptide 1aft &METIS & 603,251 & 24 & 41 &  0.004755 \\ 
		n = 384 & METIS + SA & 572,281 & 24 & 41 & 0.005007 \\ 
		m = 1833 & HMETIS & 562,601 & 24 & 40 & 0.820561 \\ 
		p = 16 & HMETIS + SA & 538,345 & 24 & 42 & 0.820771 \\ 
		&  KaHIP & 575,978      &    11      &    44  &  0.08  \\ 
		& KaHIP + SA &  575,978      &    11       &   44  &  0.08 \\
		\hline

		polyethylene chain 1024 &METIS & 8,961,763,376 & 800 & 848 &  0.01513 \\ 
		n = 12288 & METIS + SA & 8,961,763,376 & 800 & 848 & 0.017951 \\ 
		m = 290816 & HMETIS & 8,951,619,584 & 824 & 824 & 27.3297 \\ 
		p = 16 & HMETIS + SA & 8,951,619,584 & 824 & 824 & 27.3332 \\ 
		& KaHIP & 9,037,266,968     &    782     &    875    & 0.73 \\ 
		& KaHIP + SA & 9,000,224,048     &    782      &   872  &  0.74 \\
		\hline

		polyalanine 289 &METIS & 2,816,765,783,803 & 4591 & 6102 &  0.366308 \\ 
		n = 41185 & METIS + SA & 2,816,141,689,603 & 4591 & 6093 & 0.399265 \\ 
		m = 1827256 & HMETIS & 3,694,884,690,563 & 5733 & 6828 & 710.084 \\ 
		p = 16 & HMETIS + SA & 3,681,874,557,307 & 5733 & 6830 & 710.128 \\ 
		& KaHIP & 4,347,865,055,912      &    52  &      8955 &   43.9 \\ 
		& KaHIP + SA & 4,309,969,305,955    &      52    &    8907  & 43.94 \\ \hline

		peptide trp cage &METIS & 35,742,302,607 & 1228 & 1414 &  0.025795 \\ 
		n = 16863 & METIS + SA & 35,740,265,780 & 1228 & 1414 & 0.029837 \\ 
		m = 176300 & HMETIS & 35,428,817,730 & 1214 & 1472 & 31.0506 \\ 
		p = 16 & HMETIS + SA & 35,237,003,004 & 1214 & 1472 & 31.0545 \\ 
		& KaHIP & 43,551,196,287       & 515    &    1898   &  2.81\\ 
		& KaHIP + SA & 43,388,946,192      &   536      &  1896  &  2.81 \\
		\hline

		urea crystal &METIS & 4,126,744,977 & 608 & 708 &  0.047032 \\ 
		n = 3584 & METIS + SA & 4,126,744,977 & 608 & 708 & 0.057645 \\ 
		m = 109067 & HMETIS & 5,913,680,136 & 643 & 811 & 15.2321 \\ 
		p = 16 & HMETIS + SA & 5,194,749,106 & 604 & 785 & 15.2443 \\ 
		& KaHIP &  3,907,671,473 & 622 & 630 &  1.05 \\ 
		& KaHIP + SA & 3,907,671,473     &    622      &   630  &  1.05 \\
		\hline
    \end{tabular}
  \end{center}
\end{table}

\end{document}